\documentclass[onecolumn, a4size, 12pt]{IEEEtran}
\usepackage{amsmath}
\usepackage{amssymb}
\usepackage{amsfonts}
\usepackage{graphicx}
\usepackage{epsfig}
\usepackage{subfigure}
\usepackage{psfrag}

\title{Throughput Maximization in Wireless Powered Communication Networks}

\author{Hyungsik Ju and Rui Zhang
        \thanks{H. Ju is with the Department of Electrical and Computer Engineering, National University of Singapore (e-mail: elejhs@nus.edu.sg).}
       \thanks{R. Zhang is with the Department of Electrical and Computer Engineering, National University of Singapore (e-mail: elezhang@nus.edu.sg). He is also with the Institute for Infocomm Research, A*STAR, Singapore.}
       }

\begin{document}

\maketitle \thispagestyle{empty}

\setlength{\baselineskip}{1.3\baselineskip}
\newtheorem{definition}{\underline{Definition}}[section]
\newtheorem{fact}{Fact}
\newtheorem{assumption}{Assumption}
\newtheorem{theorem}{\underline{Theorem}}[section]
\newtheorem{lemma}{\underline{Lemma}}[section]
\newtheorem{corollary}{\underline{Corollary}}[section]
\newtheorem{proposition}{\underline{Proposition}}[section]
\newtheorem{example}{\underline{Example}}[section]
\newtheorem{remark}{\underline{Remark}}[section]
\newtheorem{algorithm}{\underline{Algorithm}}[section]
\newcommand{\mv}[1]{\mbox{\boldmath{$ #1 $}}}

\begin{abstract}
This paper studies the newly emerging wireless powered communication network in which one hybrid access point (H-AP) with constant power supply coordinates the wireless energy/information transmissions to/from a set of distributed users that do not have other energy sources. A ``harvest-then-transmit'' protocol is proposed where all users first harvest the wireless energy broadcast by the H-AP in the downlink (DL) and then send their independent information to the H-AP in the uplink (UL) by time-division-multiple-access (TDMA). First, we study the sum-throughput maximization of all users by jointly optimizing the time allocation for the DL wireless power transfer versus the users' UL information transmissions given a total time constraint based on the users' DL and UL channels as well as their average harvested energy values. By applying convex optimization techniques, we obtain the closed-form expressions for the optimal time allocations to maximize the sum-throughput. Our solution reveals an interesting ``doubly near-far'' phenomenon due to both the DL and UL distance-dependent signal attenuation, where a far user from the H-AP, which receives less wireless energy than a nearer user in the DL, has to transmit with more power in the UL for reliable information transmission. As a result, the maximum sum-throughput is shown to be achieved by allocating substantially more time to the near users than the far users, thus resulting in unfair rate allocation among different users. To overcome this problem, we furthermore propose a new performance metric so-called common-throughput with the additional constraint that all users should be allocated with an equal rate regardless of their distances to the H-AP. We present an efficient algorithm to solve the common-throughput maximization problem. Simulation results demonstrate the effectiveness of the common-throughput approach for solving the new doubly near-far problem in wireless powered communication networks.
\end{abstract}


\section{Introduction}
Traditionally, energy-constrained wireless networks, such as sensor networks, are powered by fixed energy sources, e.g. batteries, which have limited operation time. Although the lifetime of the network can be extended by replacing or recharging the batteries, it may be inconvenient, costly, dangerous (e.g., in a toxic environment) or even impossible (e.g., for sensors implanted in human bodies). As an alternative solution to prolong the network's lifetime, energy harvesting has recently drawn significant interests since it potentially provides unlimited power supplies to wireless networks by scavenging energy from the environment.

In particular, radio signals radiated by ambient transmitters become a viable new source for wireless energy harvesting. It has been reported that $3.5$mW and $1$uW of wireless power can be harvested from radio-frequency (RF) signals at distances of $0.6$ and $11$ meters, respectively, using Powercast RF energy-harvester operating at $915$MHz \cite{Zungeru}. Furthermore, recent advance in designing highly efficient rectifying antennas will enable more efficient wireless energy harvesting from RF signals in the near future \cite{Vullers}. It is worth noting that there has been recently a growing interest in studying wireless powered communication networks (WPCNs), where energy harvested from ambient RF signals is used to power wireless terminals in the network, e.g., \cite{Shi}-\cite{Lee}. In \cite{Shi}, a wireless powered sensor network was investigated, where a mobile charging vehicle moving in the network is employed as the energy transmitter to wirelessly power the sensor nodes. In \cite{Huang}, the wireless powered cellular network was studied in which dedicated power-beacons are deployed in the cellular network to charge mobile terminals. Moreover, the wireless powered cognitive radio network has been considered in \cite{Lee}, where active primary users are utilized as energy transmitters for charging their nearby secondary users that are not allowed to transmit over the same channel due to strong interference. Furthermore, since radio signals carry energy as well as information at the same time, a joint investigation of simultaneous wireless information and power transfer (SWIPT) has recently drawn a significant attention (see e.g. \cite{Varshney2008}-\cite{Fouladgar} and the references therein).

In this paper, we study a new type of WPCN as shown in Fig. \ref{Fig_SystemModel}, in which one hybrid access point (H-AP) with constant power supply (e.g. battery) coordinates the wireless energy/information transmissions to/from a set of distributed users that are assumed to have no other energy sources. All users are each equipped with a rechargeable battery and thus can harvest and store the wireless energy broadcast by the H-AP. Unlike prior works on SWIPT \cite{Varshney2008}-\cite{Fouladgar}, which focused on the simultaneous energy and information transmissions to users in the downlink (DL), in this paper we consider a different setup where the H-AP broadcasts only wireless energy to all users in the DL while the users transmit their independent information using their individually harvested energy to the H-AP in the uplink (UL). We are interested in maximizing the UL throughput of the aforementioned WPCN by optimally allocating the time for the DL wireless energy transfer (WET) by the H-AP and the UL wireless information transmissions (WITs) by different users.

\begin{figure}
   \centering
   \includegraphics[width=0.7\columnwidth]{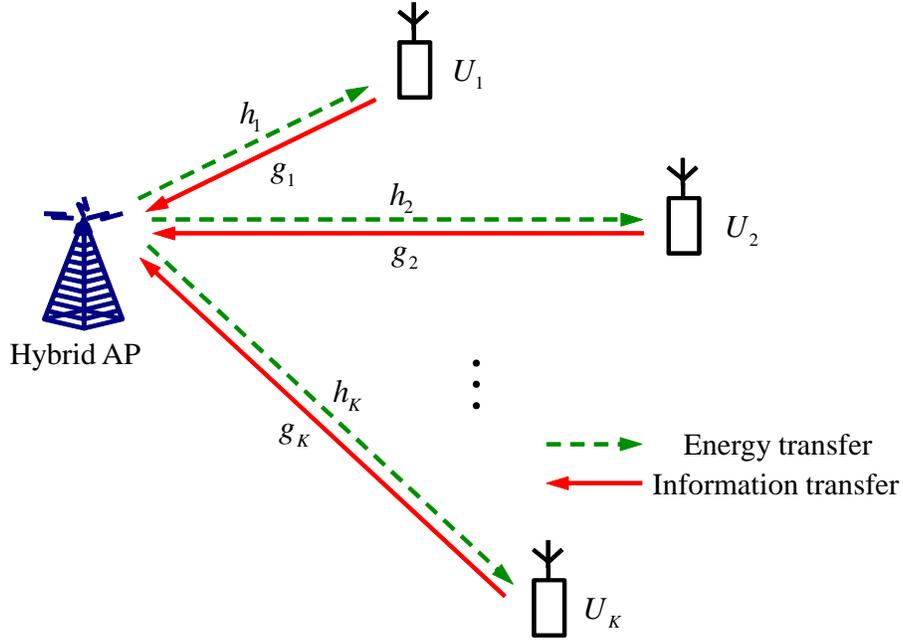}
   \caption{A wireless powered communication network (WPCN) with wireless energy transfer (WET) in the downlink (DL) and wireless information transmissions (WITs) in the uplink (UL).}
   \label{Fig_SystemModel}
\end{figure}

The main contributions of this paper are summarized as follows:

\begin{itemize}
   \item We propose a protocol termed ``harvest-then-transmit'' for the WPCN depicted in Fig. \ref{Fig_SystemModel}, where the H-AP first broadcasts wireless energy to all users in the DL, and then the users transmit their independent information to the H-AP in the UL using their individually harvested energy by time-division-multiple-access (TDMA).

   \item With the proposed protocol, we first maximize the sum-throughput of the WPCN by jointly optimizing the time allocated to the DL WET and the UL WITs given a total time constraint, based on the users' DL and UL channels as well as their average harvested energy amount. It is shown that the sum-throughput maximization problem is convex, and therefore we derive closed-form expressions for the optimal time allocations by applying convex optimization techniques \cite{Boyd}.

   \item Our solution reveals an interesting new ``doubly near-far'' phenomenon in the WPCN, when a far user from the H-AP receives less amount of wireless energy than a nearer user in the DL, but has to transmit with more power in the UL for achieving the same information rate due to the doubly distance-dependent signal attenuation in both the DL WET and UL WIT. Consequently, the sum-throughput maximization solution is shown to allocate substantially more time to the near users than the far users, thus resulting in unfair achievable rates among different users.

   \item To overcome the doubly near-far problem, we furthermore propose a new performance metric referred to as \emph{common-throughput} with the additional constraint that all users should be allocated with an equal rate in their UL WITs regardless of their distances to the H-AP. We propose an efficient algorithm to maximize the common-throughput of the WPCN by re-optimizing the time allocated for the DL WET and UL WITs. By comparing the maximum sum- versus common-throughput, we characterize the fundamental throughput-fairness trade-offs in a WPCN.
\end{itemize}

The rest of this paper is organized as follows. Section \ref{SystemModel} presents the WPCN model and the proposed harvest-then-transmit protocol. Section \ref{SumThroughputMax} studies the sum-throughput maximization problem, and characterizes the doubly near-far phenomenon. Section \ref{CommonThroughputMax} formulates the common-throughput maximization problem and presents an efficient algorithm to solve it. Section \ref{SimulationResult} presents simulation results on the sum-throughput versus common-throughput comparison. Finally, Section \ref{Conclusion} concludes the paper.

\section{System Model}\label{SystemModel}
As shown in Fig. \ref{Fig_SystemModel}, this paper considers a WPCN with WET in the DL and WITs in the UL. The network consists of one H-AP and $K$ users (e.g., sensors) denoted by $U_i$, $i = 1, \,\, \cdots, \,\, K$. It is assumed that the H-AP and all user terminals are equipped with one single antenna each. It is further assumed that the H-AP and all the users operate over the same frequency band. In addition, all user terminals are assumed to have no other embedded energy sources; thus, the users need to harvest energy from the received signals broadcast by the H-AP in the DL, which is stored in a rechargeable battery and then used to power operating circuits and transmit information in the UL.

The DL channel from the H-AP to user $U_i$ and the corresponding reversed UL channel are denoted by complex random variables $\tilde{h}_i$ and $\tilde{g}_i$, respectively, with channel power gains $h_i = {| \tilde{h}_i |^2}$ and $g_i = {| \tilde{g}_i |^2}$. It is assumed that both the DL and UL channels are quasi-static flat-fading, where $h_i$'s and $g_i$'s remain constant during each block transmission time, denoted by $T$, but can vary from one block to another. It is further assumed that the H-AP knows both $h_i$ and $g_i$, $i = 1, \,\, \cdots, \,\, K$, perfectly at the beginning of each block.

\begin{figure}
   \centering
   \includegraphics[width=0.60\columnwidth]{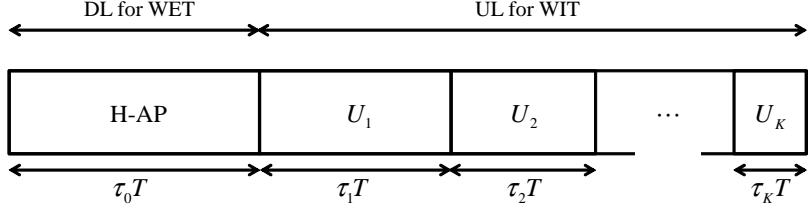}
   \caption{The harvest-then-transmit protocol.}
   \label{Fig_FrameStructure}
\end{figure}

The network adopts a \emph{harvest-then-transmit} protocol as shown in Fig. \ref{Fig_FrameStructure}. In each block, the first $\tau_0 T$ amount of time, $0 < \tau_0 <1$, is assigned to the DL for the H-AP to broadcast wireless energy to all users, while the remaining time in the same block is assigned to the UL for information transmissions, during which users transmit their independent information to the H-AP by TDMA. The amount of time assigned to user $U_i$ in the UL is denoted by $\tau_i T$, $0 \le \tau_i < 1$, $i = 1, \cdots K$. Since $\tau_0$, $\tau_1$, $\cdots$, $\tau_K$ represent the time portions in each block allocated to the H-AP and users $U_1$, $\cdots$, $U_K$ for UL WET and DL WITs, respectively, we have
\begin{equation}\label{Eq_SumTime}
   \sum\limits_{i = 0}^{K} {{\tau _i}} \le 1.
\end{equation}
For convenience, we assume a normalized unit block time $T = 1$ in the sequel without loss of generality; hence, we can use both the terms of energy and power interchangeably.

During the DL phase, the transmitted baseband signal of the H-AP in one block of interest is denoted by ${x_A}$. We assume that $x_A$ is an arbitrary complex random signal\footnote{Note that $x_A$ can also be used to send DL information at the same time; however, this usage will not be considered in this paper. Interested readers may refer to recent works on SWIPT \cite{Varshney2008}-\cite{Fouladgar}.} satisfying $\mathbb E [ {{{\left| {{x_A}} \right|}^2}} ] = P_A$, where $P_A$ denotes the transmit power at the H-AP. The received signal at $U_i$ is then expressed as
\begin{equation}\label{Eq_ReceivedSignal_Sensors}
   {{y_i} = \sqrt{h_i} {x_A} + {z_i}, \,\,\,\,\,\, i = 1, \,\, \cdots, \,\, K,}
\end{equation}
where ${y_i}$ and ${z_i}$ denote the received signal and noise at $U_i$, respectively. It is assumed that $P_A$ is sufficiently large such that the energy harvested due to the receiver noise is negligible. Thus, the amount of energy harvested by each user in the DL can be expressed as (assuming unit block time, i.e., $T = 1$)
\begin{equation}\label{Eq_HarvestedEnergy}
   {{E_i} = {\zeta_i}{P_A}{h_i}{\tau _{0}}, \,\,\,\,\,\, i = 1, \,\, \cdots, \,\, K,}
\end{equation}
where $0 < \zeta_i < 1$, $i = 1\,\, \cdots, \,\, K$, is the energy harvesting efficiency at each receiver. For convenience,  it is assumed that $\zeta_1 = \cdots = \zeta_K = \zeta$ in the sequel of this paper.

After the users replenish their energy during the DL phase, in the subsequent UL phase they transmit independent information to the H-AP in their allocated time slots. It is assumed that at each user terminal, a fixed portion of the harvested energy given by (\ref{Eq_HarvestedEnergy}) is used for its information transmission in the UL, denoted by $\eta_i$ for $U_i$, $0 < \eta_i \le 1$, $i = 1, \,\, \cdots, \,\, K$. Within $\tau_i$ amount of time assigned to $U_i$, we denote ${x_i}$ as the complex baseband signal transmitted by $U_i$, $i = 1, \,\, \cdots, \,\,K$. We assume Gaussian inputs, i.e., ${x_i} \sim {\mathcal{CN}}\left( {0,P_i} \right)$, where ${\mathcal{CN}}( {\mu ,\sigma^2})$ stands for a circularly symmetric complex Gaussian (CSCG) random variable with mean $\mu$ and variance $\sigma^2$, and $P_i$ denotes the average transmit power at $U_i$, which is given by
\begin{equation}\label{Eq_TxPower_Users}
   {P_i = \frac{\eta_i E_i}{\tau_i}, \,\,\,\,\,\, i = 1, \,\, \cdots, \,\, K.}
\end{equation}
For the purpose of exposition, we assume $\eta_i = 1$, $\forall i$, in the sequel, i.e., all the energy harvested at each user is used for its UL information transmission. The received signal at the H-AP in the $i$th UL slot is then expressed as
\begin{equation}\label{Eq_ReceivedSignal_AP}
   {{y_{A,i}} = \sqrt{g_i}{x_i} + {z_{A,i}}, \,\,\,\,\,\, i = 1, \,\, \cdots, \,\, K,}
\end{equation}
where ${y_{A,i}}$ and ${z_{A,i}}$ denote the received signal and noise at the H-AP, respectively, during slot $i$. It is assumed that ${z_{A,i}} \sim {\mathcal{CN}}\left( {0,\sigma^2} \right)$, $\forall i$. From (\ref{Eq_HarvestedEnergy})-(\ref{Eq_ReceivedSignal_AP}), the achievable UL throughput of $U_i$ in bits/second/Hz (bps/Hz) can be expressed as
\[
   {R_i}\left( {\boldsymbol{\tau}} \right) = \tau_i {\log _2}\left( {1 + \frac{{{g_i}{P_i}}}{{{\Gamma}\sigma^2}}} \right) \,\,\,\,\,\,\,\,\,\,\,\,\,\,\,\,\,\,\,\,\,\,\,\,\,\,\,\,\,
\]
\begin{equation}\label{Eq_Rate_Reference}
   {\,\,\,\,\,\,\,\,\,\,\,\,\,\,\,\,\,\,\,\,\,\,\,\,\,\,\,\,\,\,\,\,\,\,\, = \tau_i {\log _2}\left( {1 + \gamma_i\frac{\tau_{0}}{\tau_i}} \right), \,\,\,\,\,\, i = 1, \,\, \cdots, \,\, K,}
\end{equation}
where ${\boldsymbol{\tau}} = [\tau_0 \,\, \tau_1 \,\, \cdots \tau_K]$, and $\Gamma$ represents the signal-to-noise ratio (SNR) gap from the additive white Gaussian noise (AWGN) channel capacity due to a practical modulation and coding scheme (MCS) used. In addition, $\gamma_i$ is given by
\begin{equation}\label{Eq_Effective_SNR}
   {\gamma_i = \frac{\zeta{{h_i}{g_i}{P_A}}}{{{\Gamma} \sigma^2}}, \,\,\,\,\,\, i = 1, \,\, \cdots, \,\, K.}
\end{equation}

\begin{figure}[!t]
   \centering
   \includegraphics[width=0.68\columnwidth]{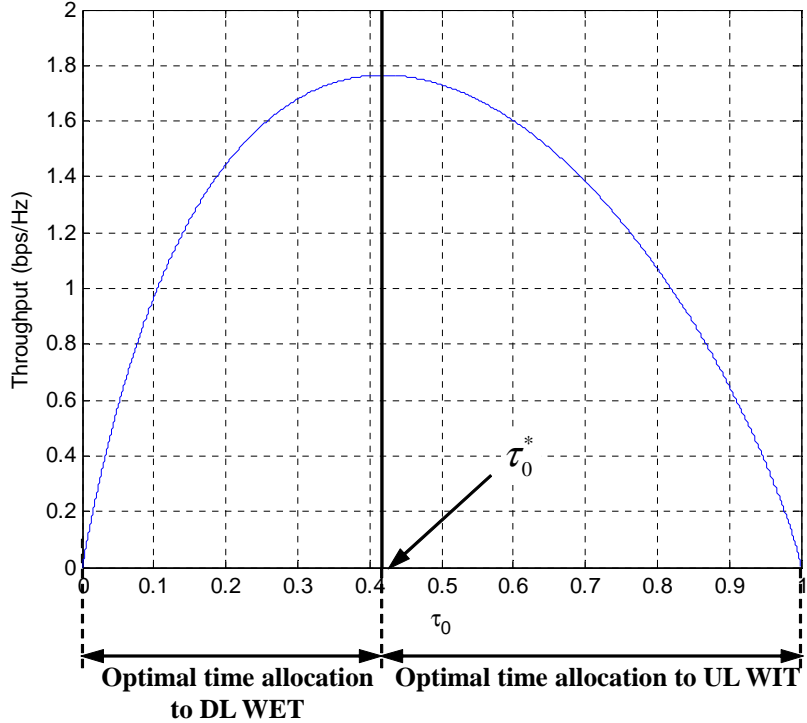}
   \caption{Throughput versus time allocated to DL WET in a single-user WPCN with $\gamma_1 = 10$dB.}
   \label{Fig_Example_Throughput_N_1}
\end{figure}

From (\ref{Eq_Rate_Reference}), it is observed that $R_i \left( \boldsymbol{\tau} \right)$ increases with $\tau_0$ for a given $\tau_i$. In addition, it can also be shown that $R_i \left( \boldsymbol{\tau} \right)$ increases with $\tau_i$ for a given $\tau_0$. However, $\tau_0$ and $\tau_i$'s cannot be increased at the same time given their total time constraint in (\ref{Eq_SumTime}). Fig. \ref{Fig_Example_Throughput_N_1} shows the throughput given in (\ref{Eq_Rate_Reference}) for the special case of one single user in the network, i.e., $K = 1$, versus the time allocated to the DL WET, $\tau_0$, with $\gamma_1 = 10$dB, assuming that (\ref{Eq_SumTime}) holds with equality, i.e., for the UL WIT $\tau_1 = 1 - \tau_0$. It is observed that the throughput is zero when $\tau_0 = 0$, i.e., no time is assigned for WET to the user in the DL and thus no energy is available for WIT in the UL, as well as when $\tau_0 = 1$ or $\tau_1  = 1 - \tau_0 = 0$, i.e., no time is assigned to the user for WIT in the UL. It is also observed that the throughput first increases with $\tau_0$ when $\tau_0 < \tau_0^* = 0.42$, but decreases with increasing $\tau_0$ when $\tau_0 > \tau_0^*$, where $\tau_0^*$ is the optimal time allocation to maximize the throughput. This can be explained as follows. With small $\tau_0$, the amount of energy harvested by $U_1$ in the DL is small. In this regime, as $U_1$ harvests more energy with increasing $\tau_0$, i.e., more energy is available for the information transmission in the UL, the throughput increases with $\tau_0$. However, as $\tau_0$ becomes larger than $\tau_0^*$, the throughput is decreased more significantly due to the reduction in the allocated UL transmission time, $\tau_1$; as a result, the throughput starts to decrease with increasing $\tau_0$. Therefore, there exists a unique optimal $\tau_0^*$ to maximize the throughput.

\section{Sum-Throughput Maximization}\label{SumThroughputMax}
In this section, we characterize the maximum sum-throughput of the WPCN presented in Section \ref{SystemModel} with arbitrary number of users, $K$. From (\ref{Eq_Rate_Reference}), the sum-throughput of all users is given by $R_{\rm{sum}}\left( {{\boldsymbol{\tau}}} \right) = \sum\limits_{i = 1}^K R_i \left({\boldsymbol{\tau}}\right)$, which is a function of the DL and UL time allocation $\boldsymbol{\tau}$. Therefore, from (\ref{Eq_SumTime}) the sum-throughput maximization problem is formulated as
\clearpage
\[
   {({\rm{P1}}): \,\,\,\,\, \mathop {\max }\limits_{{\boldsymbol{\tau}}} \,\,\,\,\,R_{\rm{sum}}\left( {{\boldsymbol{\tau}}} \right) \,\,\,\,\,\,\,\,\,\,\,\,\,\,\,\,\,\,\,\,\,\,\,\,\,\,\,\,\,}
\]
\[
   {{\rm{s.t.}}\,\,\,\,\,\, \sum\limits_{i = 0}^K {{\tau _i}}  \le 1, }
\]
\begin{equation}\label{Eq_SumRateMaxConstraint2}
   { \,\,\,\,\,\,\,\,\,\,\,\,\,\,\,\,\,\,\,\,\,\,\,\,\,\,\,\,\,\,\,\,\,\,\,\,\,\,\,\,\,\,\,\,\,\,\,\,\,\,\,\,\,\, \tau_i  \ge 0, \,\,\,\, i = 0, \,\,1, \,\, \cdots \,\, K. }
\end{equation}

\begin{lemma}\label{Lemma_SumRateConcavity}
   $R_i\left( {\boldsymbol{\tau }} \right)$ is a concave function of $\boldsymbol{\tau}$ for any given $i \in \left\{ {1, \,\, \cdots, \,\, K} \right\}$.
\end{lemma}
\begin{proof}
   Please refer to Appendix \ref{App_Proof_Lemma_SumRateConcavity}.
\end{proof}

From Lemma \ref{Lemma_SumRateConcavity}, it follows that $R_{\rm{sum}}\left( {{\boldsymbol{\tau}}} \right)$ is also a concave function of $\boldsymbol{\tau}$ since it is the summation of $R_i \left( \boldsymbol{\tau} \right)$'s. Therefore, $(\rm{P1})$ is a convex optimization problem, and thus can be solved by convex optimization techniques. To solve ($\rm{P1}$), we first have the following lemma.

\begin{lemma}\label{Lemma_ExistenceOfSolution}
   Given $A > 0$, there exists a unique $z^* > 1$ that is the solution of $f\left( z \right) =  A$, where
   \begin{equation}\label{Eq_Lemma_Existence}
      {f\left( z \right) \buildrel \Delta \over =  z\ln z - z + 1, \,\,\, z \ge 0.}
   \end{equation}
\end{lemma}
\begin{proof}
   Please refer to Appendix \ref{App_Proof_Lemma_ExistenceOfSolution}.
\end{proof}

\begin{figure}
   \centering
   \includegraphics[width=0.63\columnwidth]{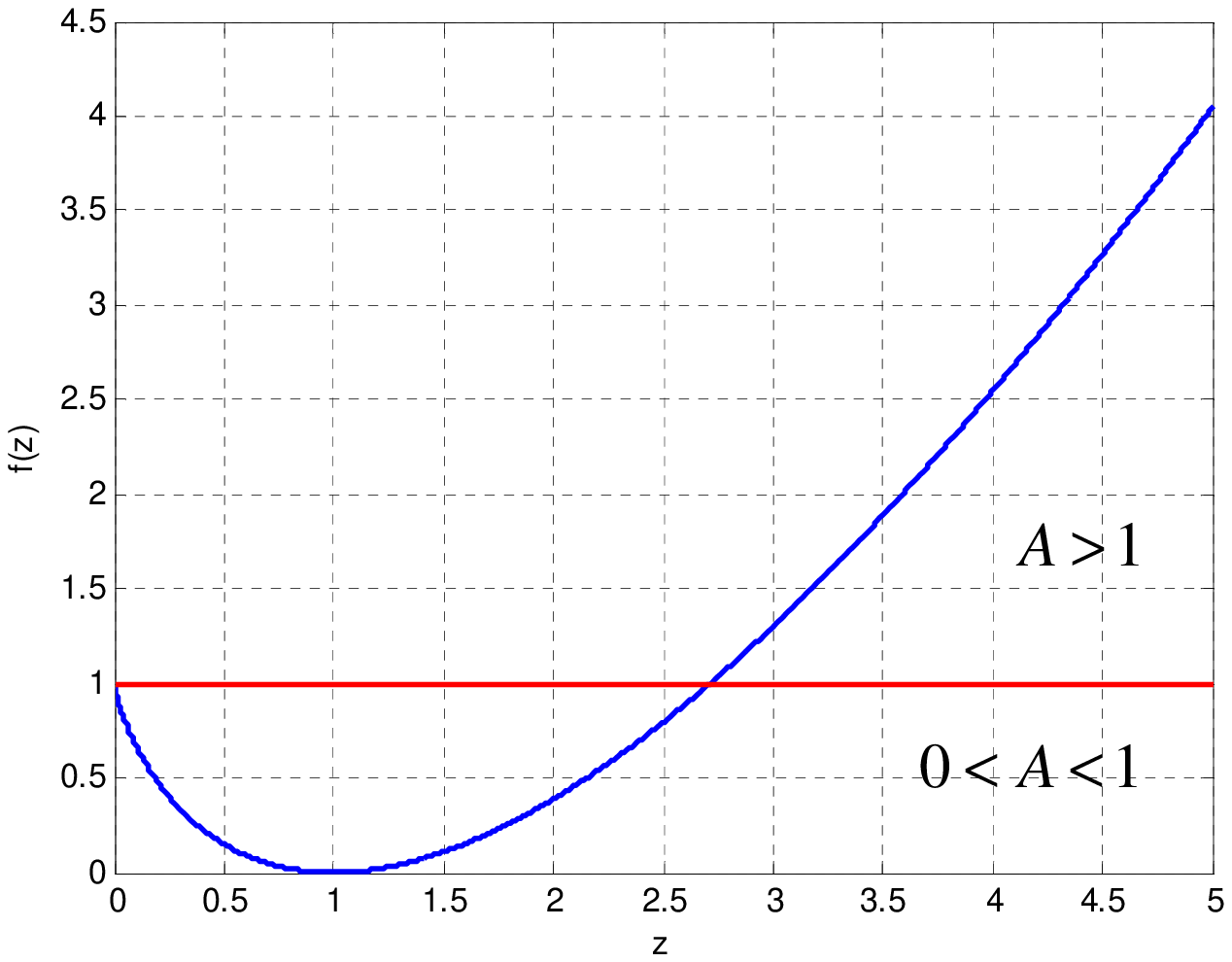}
   \caption{Plot of $f\left( z \right)$ given in (\ref{Eq_Lemma_Existence}) versus $z \ge 0$. }
   \label{Fig_Function_of_z}
\end{figure}

Fig. \ref{Fig_Function_of_z} shows $f \left( z \right)$ given in (\ref{Eq_Lemma_Existence}) with $z \ge 0$. It is observed that $f \left( z \right)$ is a convex function over $z \ge 0$ where the minimum is attained at $z = 1$ with $f \left( 1 \right) = 0$. Therefore, given $0 < A \le 1$, there are two different solutions for $f \left( z \right) = A$, among which one is smaller than $1$ and the other is larger than $1$, i.e., $z^* > 1$. On the other hand, if $A > 1$, there is only one solution for $f \left( z \right) = A$, which is larger than $1$, i.e., $z^* > 1$. The above observations are thus in accordance with Lemma \ref{Lemma_ExistenceOfSolution}.

\begin{proposition}\label{Proposition_SumRateOptimalTime}
   The optimal time allocation solution for $(\rm{P1})$, denoted by ${{\boldsymbol{\tau }}^*} = {\left[ {\tau _0^*\,\,\,\tau _1^*\,\,\, \cdots \,\,\,\tau _{K}^*} \right]}$, is given by
   \begin{equation}\label{Eq_Proposition_OptTimeAlloc}
      {\tau _i^* = \left\{ {\begin{array}{*{20}{c}}
      {\frac{{{z^* - 1}}}{{A + z^* - 1}}}  \\
      {\frac{{\gamma_i}}{{A + z^* - 1}}}  \\
      \end{array}\begin{array}{*{20}{c}}
      {,\,\,\,\,\,\,\,\,\,\,\,\,\, i = 0 \,\,\,\,\,\,\,\,\,\,}  \\
      {,\,\,\, i = 1, \cdots, K}  \\
      \end{array}} \right.}
   \end{equation}
   where ${A \buildrel \Delta \over = \sum\limits_{i = 1}^K {{\gamma_i}}} > 0$ and $z^* > 1$ is the corresponding solution of $f \left( z \right) = A$ as given by Lemma \ref{Lemma_ExistenceOfSolution}.
\end{proposition}
\begin{proof}
   Please refer to Appendix \ref{App_Proof_Proposition_SumRateOptimalTime}.
\end{proof}

It is worth noting that $A > 0$ always holds since from (\ref{Eq_Effective_SNR}) we have $\gamma_i > 0$, $i=1,\,\, 2,\,\, \cdots, K$, provided that $h_i \ne 0$ and $g_i \ne 0$. Hence, given a set of strictly positive $\gamma_i$'s, according to Lemma \ref{Lemma_ExistenceOfSolution} $z^* > 1$ is uniquely determined with the presumed $A$, thus resulting in a unique solution ${\boldsymbol{\tau}}^*$ for ($\rm{P1}$), with $\tau_i^* > 0$, $i = 0, \,\, 1, \,\, \cdots, \,\, K$, i.e., the time allocated to the DL WET is always greater than zero, and so is the time allocated to each user in the UL WIT, provided that $\gamma_i > 0$, $\forall i$. Furthermore, from Proposition \ref{Proposition_SumRateOptimalTime} we have the following corollary.

\begin{corollary}\label{Corollary_SumThroughputMax_tau0}
   In the optimal time allocation solution of ($\rm{P1}$), $\tau_0^*$ is a monotonically decreasing function of $A > 0$.
\end{corollary}
\begin{proof}
   Please refer to Appendix \ref{App_Proof_Corollary_SumThroughputMax_tau0}.
\end{proof}

From Corollary \ref{Corollary_SumThroughputMax_tau0}, it is inferred that the time allocated to the DL WET decreases with increasing $\gamma_i$'s, or channel power gains $h_i$'s and/or $g_i$'s, since ${A = \sum\limits_{i = 1}^K {{\gamma_i}}}$ and $\gamma_i \propto h_i g_i$, $i = 1, \,\, \cdots, K$, as shown in (\ref{Eq_Effective_SNR}). As a result, $\tau_i ^*$'s, $i = 1, \,\, \cdots, \,\, K$, increase with $A$, i.e., the time allocated to the UL WIT increases with $\gamma_i$'s. This is an interesting observation implying that when the channel power gains, $h_i$'s and $g_i$'s, become larger, we should allocate more time to the UL WITs instead of the DL WET to maximize the sum-throughput. This is because with larger $\gamma_i$'s, the required energy for UL WITs becomes smaller given any transmission rate; thus, each user can harvest sufficient amount of wireless energy from the H-AP even with a smaller time allocated to the DL WET.

Note that the sum-throughput maximization solution given in (\ref{Eq_Proposition_OptTimeAlloc}) allocates more time to a `near' user to the H-AP than a `far' user, since in practice $h_i \propto D_{i}^{-\alpha_d}$, $g_i \propto D_{i}^{-\alpha_u}$, and $\gamma_i \propto h_i g_i$ according to (\ref{Eq_Effective_SNR}), where $\alpha_d \ge 2$ and $\alpha_u \ge 2$ denote the channel pathloss exponents in the DL and UL, respectively, and $D_i$ denotes the distance between the H-AP and $U_i$. Thus, from (\ref{Eq_Proposition_OptTimeAlloc}) it follows that $\tau_i^* \propto D_i^{-(\alpha_d + \alpha_u)}$, $i = 1, \,\, \cdots, K$, which results in an unfair time and throughput allocation among users in the WPCN, a phenomenon termed ``doubly near-far problem''. To further illustrate this issue, Fig. \ref{Fig_Example_SumThroughput_N_2} shows the sum-throughput $R_{\rm{sum}}\left( {{\boldsymbol{\tau}}} \right)$ versus UL WIT time allocation $\tau_1$ and $\tau_2$ for a two-user network with $K = 2$ and $D_1 = \frac{1}{2}D_2$. It is assumed that the channel reciprocity holds for the DL and UL and thus $h_i = g_i$, $i = 1, 2$, with $\alpha_{d} = \alpha_u = 2$. Accordingly, we set $\gamma_1 = 22$dB and $\gamma_2 = 10$dB, with $\gamma_1/\gamma_2 = \left( D_2/D_1 \right)^{\alpha_d + \alpha_u}$. It is observed from Fig. \ref{Fig_Example_SumThroughput_N_2} that $R_{\rm{sum}}\left( {{\boldsymbol{\tau}}} \right) = 0$ when $\tau_1 = \tau_2 = 0$, $\tau_0 = 1 - \left( \tau_1 + \tau_2 \right) = 1$, or $\tau_1 + \tau_2 = 1$, $\tau_0 = 0$, since no time is allocated to the users for the UL WITs in the former case, while no time is allocated to the DL WET in the latter case. The numerical result of sum-throughput clearly shows that $R_{\rm{sum}}\left( {{\boldsymbol{\tau}}} \right)$ is strictly positive when $0 < \tau_1 + \tau_2 <1$. In addition, it is observed that the optimal DL and UL time allocation to maximize the sum-throughput is ${\boldsymbol{\tau}}^* = [0.2441, \,\, 0.7114, \,\, 0.0445]$ where $\tau_1^* = 16 \tau_2^*$, i.e., $\tau_1^* = (D_2/D_1)^{\alpha_{d} + \alpha_{u}} \tau_2^*$, which is consistent with (\ref{Eq_Proposition_OptTimeAlloc}). Furthermore, at the optimal ${\boldsymbol{\tau}}^*$, $R_{\rm{sum}}\left( {{\boldsymbol{\tau}}}^* \right) = 4.58$ bps/Hz, with $R_1\left( {{\boldsymbol{\tau}}^*} \right) = 4.13$ bps/Hz and $R_2\left( {{\boldsymbol{\tau}}^*} \right) = 0.45$ bps/Hz, which demonstrates the very unfair throughput allocation between the two users due to the doubly near-far problem.

\begin{figure}
   \centering
   \includegraphics[width=0.66\columnwidth]{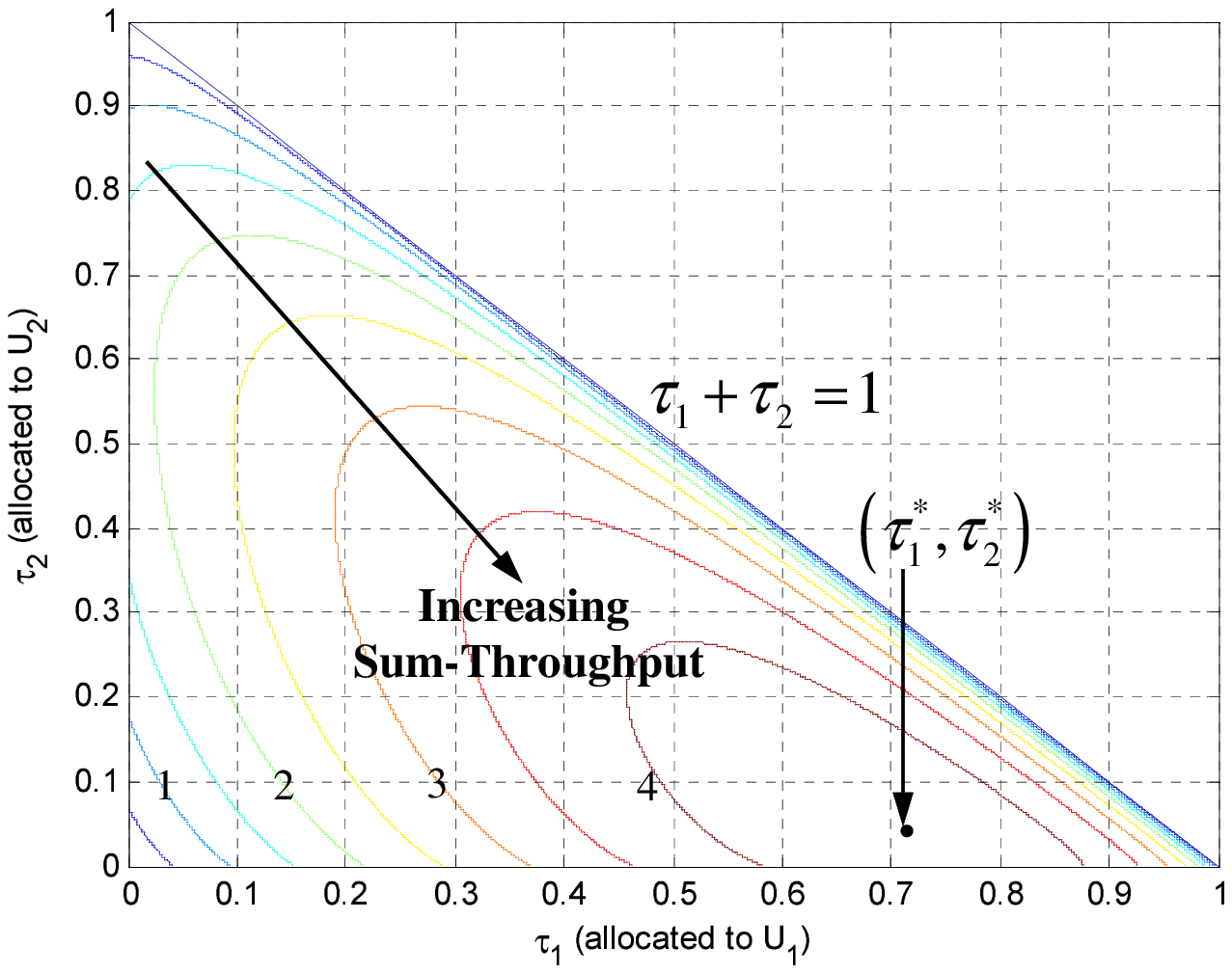}
   \caption{Sum-throughput (in bps/Hz) versus time allocation. }
   \label{Fig_Example_SumThroughput_N_2}
\end{figure}

For comparison, we consider UL transmissions in a conventional TDMA-based wireless network with WIT only \cite{Knopp}-\cite{Zhang_MAC}, where each user is equipped with a constant energy supply, and thus has an equal energy consumption at each block denoted by ${\bar{E}}$. It then follows from (\ref{Eq_Effective_SNR}) that $\gamma_i \propto D_i^{-\alpha_u}$ and thus from (\ref{Eq_Proposition_OptTimeAlloc}), the optimal time allocation to maximize the sum-throughput of such a conventional TDMA network should satisfy $\tau_i^* \propto D_i^{-\alpha_u}$, $i = 1, \,\, \cdots, \,\, K$, and $\tau_0^* = 0$ (since no DL WET is needed). Clearly, the WPCN suffers from a more severe near-far problem than the conventional TDMA network. With the same setup as for Fig. \ref{Fig_Example_SumThroughput_N_2}, in Fig. \ref{Fig_WPCN_vs_TDMA_N_2} we show the optimal throughput of $U_2$, $R_2 \left( {\boldsymbol{\tau}}^* \right)$, normalized by that of $U_1$, $R_1 \left( {\boldsymbol{\tau}}^* \right)$, in a WPCN versus that in a conventional TDMA network for different values of the pathloss exponent $\alpha$, with $\alpha_d = \alpha_u = \alpha$. For the WPCN, $\gamma_1$ is set to be fixed as $\gamma_1 = 22$dB for $U_1$, while $\gamma_2 = 10$, $7$, $4$, $1$, and $-2$dB for $\alpha = 2$, $2.5$, $3$, $3.5$, and $4$, respectively, since $\gamma_1/\gamma_2 = \left( D_2/D_1 \right)^{2\alpha}$. For the conventional TDMA network, $\bar E$ for both $U_1$ and $U_2$ are assumed to be $\bar E = \frac{1}{2}\left( E_1 + E_2 \right)$ with $E_1$ and $E_2$ denoting the average harvested energy at $U_1$ and $U_2$ in the WPCN under comparison, respectively; it then follows that for the conventional TDMA network, $\gamma_1 = 13$dB and $\gamma_2 = 7.1$, $5.5$, $4.0$, $2.4$, and $0.9$dB for $\alpha = 2$, $2.5$, $3$, $3.5$, and $4$, respectively. From Fig. \ref{Fig_WPCN_vs_TDMA_N_2}, it is observed that the throughput ratio of the two users in the WPNC case decreases twice faster than that in the conventional TDMA in the logarithm scale due to the more severe (doubly) near-far problem.

\begin{figure}
   \centering
   \includegraphics[width=0.66\columnwidth]{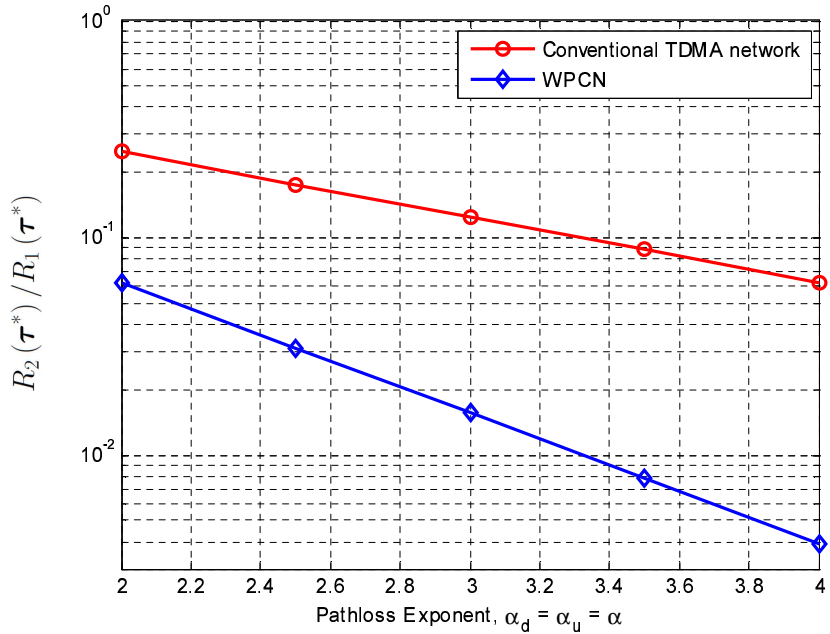}
   \caption{Throughput ratio in a two-user WPCN versus conventional TDMA network with equal energy supply. }
   \label{Fig_WPCN_vs_TDMA_N_2}
\end{figure}

\section{Common-Throughput Maximization}\label{CommonThroughputMax}
In this section, we tackle the doubly near-far problem in the WPNC by applying the common-throughput maximization approach, which guarantees equal throughput allocations to all users and yet maximize their sum-throughput. From (\ref{Eq_SumTime}) and (\ref{Eq_Rate_Reference}), the common-throughput maximization problem is formulated as
\[
   {({\rm{P}}{\rm{2}}):\,\,\, \mathop {\max }\limits_{\bar{R},\boldsymbol{\tau} } \,\,\,\, \bar{R}\,\,\,\,\,\,\,\,\,\,\,\,\,\,\,\,\,\,\,\,\,\,\,\,\,\,\,\,\,\,\,\,\,\,\,\,\,\,\,\,\,\,\,\,\,\,\,\,\,\,\,\,\,\,\,\,\,\,\,\,\,\,\,\,\,\,\,\,}
\]
\begin{equation}\label{Eq_CommonRateMaxConstraint1}
   {\,\,\,\,\,\,\,\,\,\,\,\,\,\, {\rm{s.t.}}\,\,\,\,\, R_i\left( {\boldsymbol{\tau }} \right) \ge \bar{R}, \,\,\,i = 1 \cdots K,}
\end{equation}
\[
   {\boldsymbol{\tau}  \in \mathcal{D}, \,\,\,\,\,\,\,\,\,\,\,\,\,\,\,\,}
\]
where $\bar R$ denotes the common-throughput and $\mathcal{D}$ is the feasible set of $\boldsymbol{\tau}$ specified by (\ref{Eq_SumTime}) and (\ref{Eq_SumRateMaxConstraint2}).

\begin{remark}\label{Remark_EqualRate}
   Problem $({\rm{P}}{\rm{2}})$ is designed to guarantee the throughput of the user with the worst channel condition, e.g., of the largest distance from the H-AP. Since $R_i\left( {\boldsymbol{\tau }} \right)$ given by (\ref{Eq_Rate_Reference}) is a monotonically increasing function of both $\tau_0$ and $\tau_i$, it can be easily shown that the optimal time allocation solution ${\boldsymbol{\tau}}^*$ for ($\rm{P2}$) should allocate the same optimal throughput to all the users, denoted by ${\bar{R}^{*}} = R_1 \left( {\boldsymbol{\tau}}^{*} \right) = \cdots = R_K \left( {\boldsymbol{\tau}}^{*} \right)$, with $\sum\limits_{i = 0}^K {\tau _i^*}  = 1$, when the minimum user throughput in the network is maximized. In addition, allocating equal throughput to all users can be relevant in practice, since one typical application of the WPCN is sensor network, where all the sensors may need to periodically send their sensing data to a function center (modelled as the H-AP in our setup) with the same rate.
\end{remark}

The maximum common-throughput ${\bar{R}}^*$ is the maximum of all the feasible common-throughput $\bar R$ that satisfies the rate inequalities in (\ref{Eq_CommonRateMaxConstraint1}) of ($\rm{P2}$). To solve ($\rm{P2}$), given any $\bar R > 0$, we first consider the following feasibility problem:
\[
   {{\rm{Find}} \,\,\,\,\, {{\boldsymbol{\tau}}} \,\,\,\,\,\,\,\,\,\,\,\,\,\,\,\,\,\,\,\,\,\,\,\,\, }
\]
\[
   {\,\,\,\,\,\,\,\,\,\,\,\,\,\,\,\,\,\,\,\,\,\,\,\,\,\,\,\,\,\,\,\,\,\,\,\,\,\,\,\,\,\,\,\, {\rm{s.t.}}\,\,\,\,\,\,{R_i}\left( \boldsymbol{\tau}  \right)  \ge \bar{R},\,\,\, i = 1, \,\, \cdots, K,}
\]
\begin{equation}\label{Eq_FeasibilityProblem}
   { \,\,\,\,\,\,\, \boldsymbol{\tau}  \in \mathcal{D}. }
\end{equation}
Since the problem in (\ref{Eq_FeasibilityProblem}) is convex, we consider its Lagrangian given by
\begin{equation}\label{Eq_Lagrangian}
   {{\mathcal{L}}\left( {\boldsymbol{\tau} ,\boldsymbol{\lambda} } \right) =  - \sum\limits_{i = 1}^K {{\lambda _i}\left( { R_i \left( \boldsymbol{\tau} \right) - \bar R} \right)} ,}
\end{equation}
where $\boldsymbol{\lambda} = [\lambda_1, \,\, \cdots, \lambda_K] \ge 0$ (`$\ge$' denotes the component-wise inequality) consists of the Lagrange multipliers associated with the $K$ user throughput constraints in problem (\ref{Eq_FeasibilityProblem}). The dual function of problem (\ref{Eq_FeasibilityProblem}) is then given by
\begin{equation}\label{Eq_DualFunction}
   {{\mathcal{G}}\left( \boldsymbol{\lambda}  \right) = \mathop {\min }\limits_{\boldsymbol{\tau}  \in \mathcal{D}} {\mathcal{L}}\left( {\boldsymbol{\tau} ,\boldsymbol{\lambda} } \right).}
\end{equation}
The dual function ${\mathcal{G}}\left( \boldsymbol{\lambda}  \right)$ can be used to determine whether problem (\ref{Eq_FeasibilityProblem}) is feasible, as provided in the following lemma.

\begin{lemma}\label{Lemma_DualFuntion_Feasibility}
   For a given $\bar R > 0$, problem (\ref{Eq_FeasibilityProblem}) is infeasible if and only if there exists an ${\boldsymbol{\lambda}} \ge 0$ such that $\mathcal{G}(\boldsymbol{\lambda}) > 0$.
\end{lemma}
\begin{proof}
   Please refer to Appendix \ref{App_Proof_Lemma_DualFunction_Feasibility}.
\end{proof}

Next, we obtain $\mathcal{G} ( \boldsymbol{\lambda})$ in (\ref{Eq_DualFunction}) for a given $\boldsymbol{\lambda} \ge 0$ by solving the following weighted sum-throughput maximization problem, which follows from (\ref{Eq_Lagrangian}).
\[
   \mathop {\max }\limits_{\boldsymbol{\tau}}  \,\,\,\,\,\sum\limits_{i = 1}^K {{\lambda _i}{R_i}\left( \boldsymbol{\tau}  \right)\,\,\,}
\]
\begin{equation}\label{Eq_WSRProblem}
   {{\rm{s.t.}}\,\,\,\,\,\, \boldsymbol{\tau}  \in \mathcal{D}. \,\,\,\,\,\,\,\,\,\,\,\,\,\,}
\end{equation}

Like ($\rm{P1}$), the problem in (\ref{Eq_WSRProblem}) is convex and thus can be solved by convex optimization techniques. Similar to Proposition \ref{Proposition_SumRateOptimalTime} for the sum-throughput maximization case with $\lambda_i = 1$, $\forall i$, we obtain the optimal time allocation solution for the weighted sum-throughput maximization problem in (\ref{Eq_WSRProblem}), given in the following proposition.

\begin{proposition}\label{Proposition_WeightedSumRate}
   Given $\boldsymbol{\lambda} \ge 0$, the optimal time allocation solution for (\ref{Eq_WSRProblem}), denoted by ${\boldsymbol{\tau}}^{\star}$ $=$ $\left[ {{\tau}_0^{\star},\,\, {\tau}_1^{\star},} \right.$ $\left. {\cdots, {\tau}_K^{\star}} \right]$, is
   \begin{equation}\label{Eq_WSR_TimeAlloc_DL}
      {{\tau} _0^{\star} = \frac{1}{{1 + \sum\limits_{j = 1}^K { \left( \gamma_j / {{z}_j^{\star}} \right) } }},}
   \end{equation}
   \begin{equation}\label{Eq_WSR_TimeAlloc_UL}
      {{\tau} _i^{\star} = \frac{ \gamma_i / {{z}_i^{\star}} }{{1 + \sum\limits_{j = 1}^K { \left( \gamma_j / {{z}_j^{\star}} \right) } }}, \,\,\, i = 1, \,\, \cdots, \,\, K,}
   \end{equation}
   where ${z}_i^{\star}$, $i = 1, \,\, \cdots, \,\, K$, is the solution of the following equations:
   \begin{equation}\label{Eq_WSR_KKT1}
      {\ln \left( {1 + {z_i}} \right) - \frac{{{z_i}}}{{1 + {z_i}}} = \frac{\mu^{\star}}{{{\lambda _i}}}\ln 2, }
   \end{equation}
   \begin{equation}\label{Eq_WSR_KKT2}
      {\sum\limits_{i = 1}^K {\frac{{{\lambda _i}{\gamma_i}}}{{1 + {z_i}}}}  = {{{\mu}^{\star}}}\ln 2,}
   \end{equation}
   with ${\mu}^{\star} > 0$ being a constant.
\end{proposition}
\begin{proof}
   Please refer to Appendix \ref{App_Proof_Lemma_DualFuntion_Feasibility}.
\end{proof}

With Proposition \ref{Proposition_WeightedSumRate}, we can compute ${\boldsymbol{\tau}}^{\star}$ efficiently as follows. Denote the left-hand sides (LHSs) of (\ref{Eq_WSR_KKT1}) and (\ref{Eq_WSR_KKT2}) as $Q \left( z_i \right)$ and $S \left( \bf{z} \right)$ with ${\bf{z}} = [z_1, \,\, z_2, \,\, \cdots, \,\, z_K]$, respectively. Note that $Q \left( z_i \right)$ is an increasing function of $z_i$, $i=1, \,\, \cdots, K$ (see Appendix \ref{App_Proof_Proposition_SumRateOptimalTime}), whereas $S \left( \bf{z} \right)$ is a decreasing function with respect to each individual $z_i$. Given any $\mu > 0$, suppose that $z_i$ is the solution of $Q \left( z_i \right) = \frac{\mu }{{{\lambda _i}}}\ln 2$, $i=1, \,\, \cdots, K$, in (\ref{Eq_WSR_KKT1}). With these $z_i$'s, there are two possible cases to consider next. If in (\ref{Eq_WSR_KKT2}) the resulting $S \left( \bf{z} \right) > \mu \ln 2$, we should increase $\mu$ since $z_i$'s satisfying $Q \left( z_i \right) = \frac{\mu }{{{\lambda _i}}}\ln 2$, $i=1, \,\, \cdots, K$, will increase with $\mu$ given that $Q \left( z_i \right)$, $\forall i$, is an increasing function of $z_i$; as a result, $S \left( \bf{z} \right)$ will decrease since it is a decreasing function of each individual $z_i$. Otherwise, $\mu$ should be decreased to satisfy (\ref{Eq_WSR_KKT2}) if $S \left( \bf{z} \right) < \mu \ln 2$. Therefore, ${z}_i^{\star}$'s and $\mu^{\star}$ can be obtained by iteratively updating $z_i$'s and $\mu$ as above until convergence is reached. Then, ${\boldsymbol{\tau}}^{\star}$ can be computed from (\ref{Eq_WSR_TimeAlloc_DL}) and (\ref{Eq_WSR_TimeAlloc_UL}) accordingly.

\begin{table}[!t]
\renewcommand{\arraystretch}{1.3}
\caption{Algorithm to solve $(\rm{P2}).$}
\label{Table_Algorithm} \centering
   \begin{tabular}{|p{3.4in}|}
   \hline
      1)  \textbf{Initialize} $R_{\rm{min}} = 0$, $R_{\rm{max}} > \bar{R}^{*}$.\footnotemark

      2) \textbf{Repeat}
         \begin{itemize}
               \item[1.] $\bar{R} = \frac{1}{2} \left( R_{\rm{min}} + R_{\rm{max}} \right)$.

               \item[2.] Initialize $\boldsymbol{\lambda} \ge 0$.

               \item[3.] Given $\boldsymbol{\lambda}$, solve the problem in (\ref{Eq_WSRProblem}) by Proposition \ref{Proposition_WeightedSumRate}.

               \item[4.] Compute $\mathcal{G} \left( {\boldsymbol{\lambda}} \right)$ using (\ref{Eq_Lagrangian}).

               \item[5.] If ${\mathcal{G}} \left( {\boldsymbol{\lambda}} \right) > 0$, $\bar R$ is infeasible, set ${R_{\max }} \leftarrow \bar R$, go to step 1. \

                   Otherwise, update $\boldsymbol \lambda$ using the ellipsoid method and the subgradient of ${\mathcal{G}} \left( {\boldsymbol{\lambda}} \right)$ given by (\ref{Eq_Subgradient}). If the stopping criteria of the ellipsoid method is not met, go to step 3.

               \item[6.] Set ${R_{\min }} \leftarrow \bar R$.\

         \end{itemize}

      3) \textbf{Until} ${R_{\max }} - {R_{\min }} < \delta$, where $\delta > 0$ is a given error tolerance.
            \\
   \hline
   \end{tabular}
\end{table}
\footnotetext{The initial value of $R_{\max}$ can be chosen as any arbitrary large number such that it satisfies $R_{\max} > \bar{R}^{*}$.}

Given $\bar R$, $\boldsymbol{\lambda}$, and the obtained $\boldsymbol{\tau}^{\star}$ by solving problem (\ref{Eq_WSRProblem}) with Proposition \ref{Proposition_WeightedSumRate}, we can compute the corresponding $R_i \left( \boldsymbol{\tau}^{\star} \right)$, $i=1, \,\, \cdots, \,\, K$, and thus ${\mathcal{G}}\left( \boldsymbol{\lambda}  \right)$ in (\ref{Eq_DualFunction}) using (\ref{Eq_Lagrangian}). If ${\mathcal{G}}\left( \boldsymbol{\lambda}  \right) > 0$, it follows from Lemma \ref{Lemma_DualFuntion_Feasibility} that problem (\ref{Eq_FeasibilityProblem}) is infeasible, i.e., $\bar R > \bar R^{*}$. Therefore, we should decrease $\bar R$ and solve the feasibility problem in (\ref{Eq_FeasibilityProblem}) again. On the other hand, if ${\mathcal{G}}\left( \boldsymbol{\lambda}  \right) \le 0$, we can update $\boldsymbol{\lambda}$ using sub-gradient based algorithms, e.g. the ellipsoid method \cite{LectureNote}, with the subgradient of ${\mathcal{G}} \left( {\boldsymbol{\lambda}} \right)$, denoted by ${\boldsymbol{\upsilon}} = {[ {\upsilon _1\,\,\,\upsilon _2\,\,\, \cdots \,\,\,\upsilon _K} ]^T}$, given by
\begin{equation}\label{Eq_Subgradient}
   {\upsilon _i = \tau _i^{\star}{\log _2}\left( {1 + \frac{{{\gamma _i}\tau _0^{\star}}}{{\tau _i^{\star}}}} \right) - \bar{R}, \,\,\, 1 \le i \le K,}
\end{equation}
until $\boldsymbol{\lambda}$ converges to $\boldsymbol{\lambda}^{*}$ with $\boldsymbol{\lambda}^{*}$ denoting the maximizer of ${\mathcal{G}}\left( \boldsymbol{\lambda}  \right)$ or the optimal dual solution for problem (\ref{Eq_FeasibilityProblem}). If ${\mathcal{G}}\left( \boldsymbol{\lambda}^{*}  \right) \le 0$, it then follows that problem (\ref{Eq_FeasibilityProblem}) is feasible and thus $\bar{R} \le \bar{R}^*$. In this case, $\bar R$ should be increased for solving the feasibility problem in (\ref{Eq_FeasibilityProblem}) again. Consequently, $\bar{R}^*$ can be obtained numerically by iteratively updating $\bar R$ by a simple bisection search \cite{Boyd}. To summarize, one algorithm to solve $(\rm{P2})$ is given in Table \ref{Table_Algorithm}.\footnote{The computational complexity of the proposed algorithm in Table \ref{Table_Algorithm} can be shown to be $\mathcal{O}(K^3)$ since at each iteration it performs $K$ one-dimension searches each with the complexity of $\mathcal{O}(1)$ to find $\boldsymbol{\tau}^{\star}$, and the ellipsoid method has the complexity of $\mathcal{O}(K^2)$ \cite{LectureNote} to converge.}

Fig. \ref{Fig_Example_MinThroughput_N_2} shows the common-throughput in bps/Hz versus $\tau_1$ and $\tau_2$ for the same two-user channel setup as for Fig. \ref{Fig_Example_SumThroughput_N_2}. It is observed that the optimal time allocation for (\rm{P2}) is given by $\boldsymbol{\tau}^* = [0.3683, \,\,\,\, 0.1386, \,\,\,\, 0.4932]$, which results in $\bar{R}^{*} = R_1 \left( {\boldsymbol{\tau}}^{*} \right) = R_2 \left( {\boldsymbol{\tau}}^{*} \right) = 1.46$bps/Hz. Comparing to Fig. \ref{Fig_Example_SumThroughput_N_2} where the sum-throughput is maximized, the time portion allocated to the near user, $U_1$, is decreased substantially from $0.7114$ to $0.1737$, while that to the far user, $U_2$, is greatly increased from $0.0445$ to $0.4669$. Consequently, the throughput of $U_2$ is increased from $0.45$bps/Hz to $1.46$bps/Hz, while that of $U_1$ is decreased from $4.13$bps/Hz to $1.46$bps/Hz, the same throughput as $U_2$. This result shows the effectiveness of the proposed common-throughput approach for tackling the doubly near-far problem in a WPCN.

\begin{figure}
   \centering
   \includegraphics[width=0.7\columnwidth]{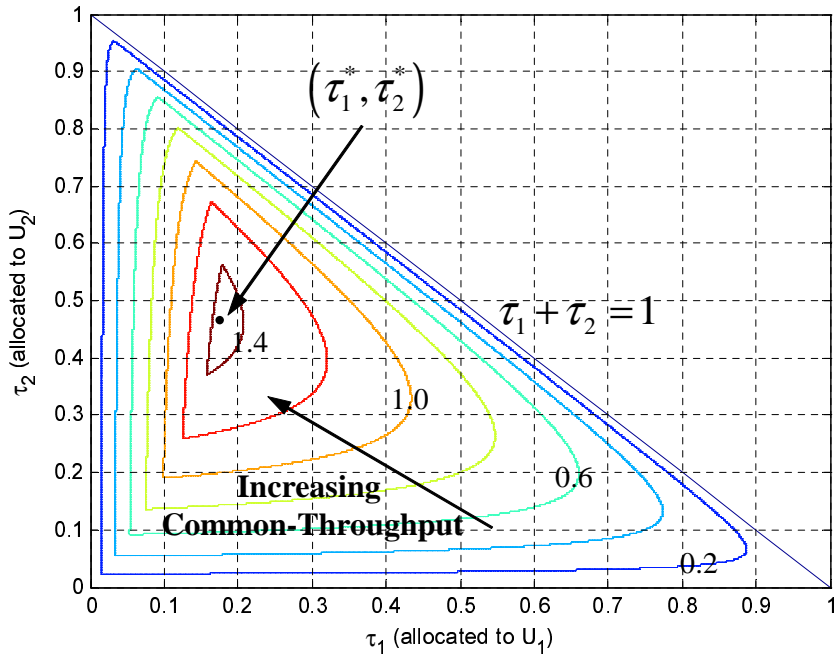}
   \caption{Common-throughput (in bps/Hz) versus time allocation. }
   \label{Fig_Example_MinThroughput_N_2}
\end{figure}

Fig. \ref{Fig_Ratio_TimeAllocation_STM_vs_CTM} shows the ratio of the optimal time allocated to the far user $U_2$ over that to the near user $U_1$, i.e., $\tau_2^* / \tau_1^*$, in ($\rm{P1}$) versus ($\rm{P2}$), with different values of the common pathloss exponent $\alpha$ in both DL and UL, where the same two-user channel setup as for Fig. \ref{Fig_Example_SumThroughput_N_2} is considered. It is assumed that $\gamma_1$ for the near user $U_1$ is fixed as $\gamma_1 = 22$dB and $\gamma_2$ for the far user $U_2$ is set the same as for Fig. \ref{Fig_WPCN_vs_TDMA_N_2}. It is observed that the time ratio of ($\rm{P1}$) in the logarithm scale decreases linearly with $\alpha$ to maximize the sum-throughput, which can also be inferred from (\ref{Eq_Proposition_OptTimeAlloc}), due to the doubly near-far problem. On the contrary, $\tau_2^* / \tau_1^*$ is observed to increase with $\alpha$ to maximize the common-throughput in ($\rm{P2}$), since more time is allocated to the far user $U_2$ instead of the near user $U_1$ as the ratio between $\gamma_1$ and $\gamma_2$, i.e., $\gamma_1/\gamma_2$, increases with $\alpha$.

\begin{figure}
   \centering
   \includegraphics[width=0.7\columnwidth]{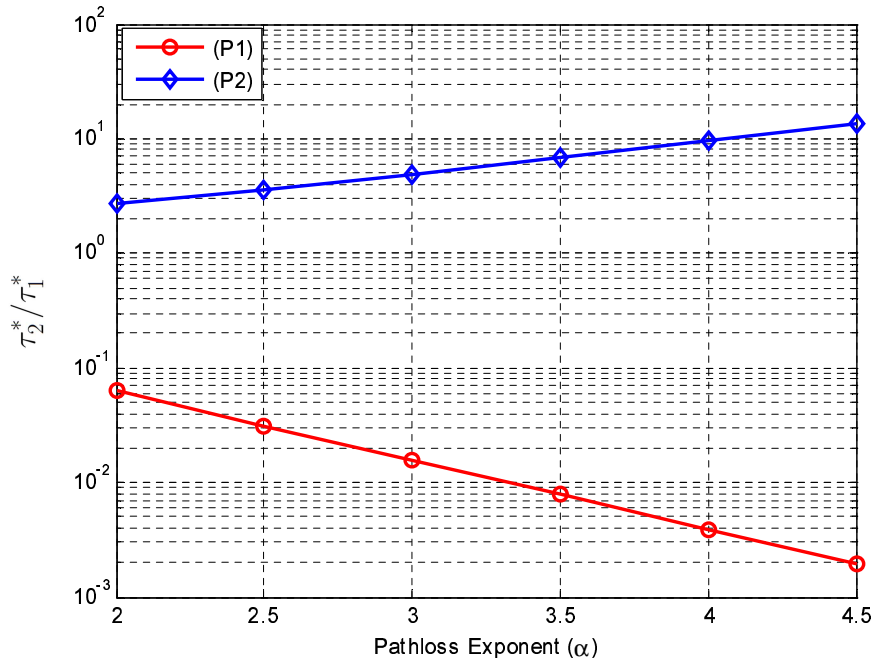}
   \caption{Comparison of the ratio of time allocated to $U_2$ and $U1$ in ($\rm{P1}$) versus ($\rm{P2}$). }
   \label{Fig_Ratio_TimeAllocation_STM_vs_CTM}
\end{figure}

Notice that ($\rm{P1}$) and ($\rm{P2}$) deal with two extreme cases of throughput allocation to the users in a WPNC where the fairness is completely ignored and a strict equal fairness is imposed, respectively. More generally, Fig. \ref{Fig_Throughput_Region} shows the achievable throughput region of a two-user WPCN by solving the weighted sum-throughput maximization problem in (\ref{Eq_WSRProblem}) with different throughput weights for the near and far users, under the same channel setup as for Fig. \ref{Fig_Example_SumThroughput_N_2}. It is observed that the boundary of the throughput region characterizes all the optimal throughput-fairness trade-offs in this two-user WPCN, which include the throughput pairs obtained by solving ($\rm{P1}$) for the maximum sum-throughput and by solving ($\rm{P2}$) for the maximum common-throughput, shown as points ($\rm{a}$) and ($\rm{b}$) in the figure, respectively.

\begin{figure}
   \centering
   \includegraphics[width=0.7\columnwidth]{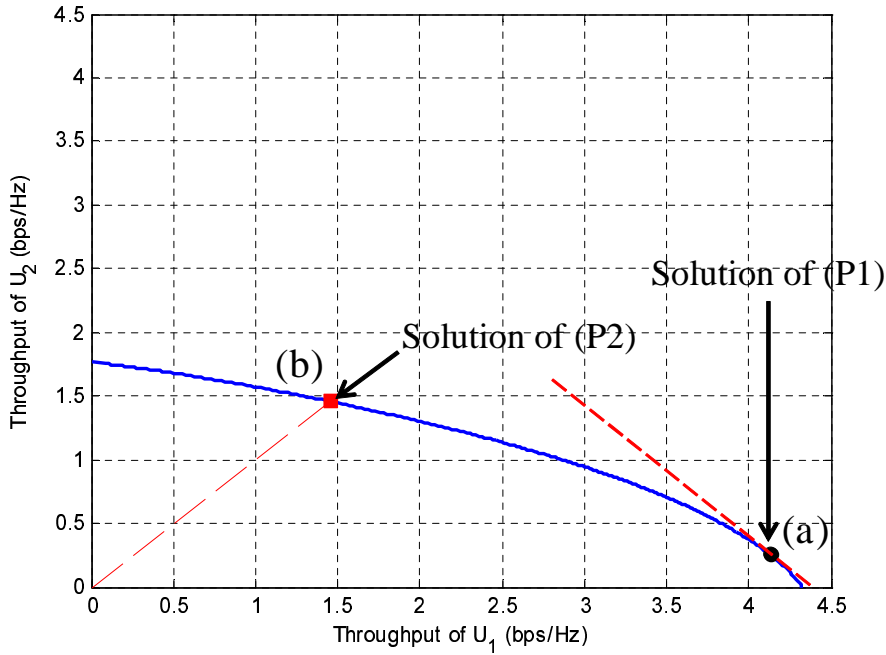}
   \caption{Throughput region of a two-user WPCN with $(\rm{a})$ corresponding to the maximum sum-throughput; and $(\rm{b})$ corresponding to the  maximum common-throughput.}
   \label{Fig_Throughput_Region}
\end{figure}

\begin{remark}\label{Remark_DifferentRate}
   It is worth noting that the common-throughput approach for characterizing the achievable rate region of multi-user communication systems under strict fairness constraints can be considered as one special case of the ``rate-profile'' method proposed in \cite{Mohseni}. Hence, the common-throughput approach investigated in this paper can be easily extended to the general case where the required throughput of each user is different using the rate-profile method. Given ${\bf{\bar R}} = \left[ {{{\bar R}_1}\,\,\,{{\bar R}_2}\,\, \cdots \,\,{{\bar R}_K}} \right]$ with ${\bar R_i}$ denoting the required throughput of user $i$, $i = 1, \,\, \cdots, \,\, K$, the corresponding rate profile vector is defined as ${\boldsymbol{\beta }} = \left[ {{\beta _1}\,\,{\beta _2}\,\, \cdots \,\,{\beta _K}} \right]$ where ${\beta _i} = {\bar R_i}/\sum\limits_{j = 1}^K {{{\bar R}_j}}$ (Note that the common-throughput maximization problem ($\rm{P2}$) is thus for a special case with $\beta_i = 1/K, \forall i$). The optimal time allocation solution to maximize the system sum-throughput subject to the rate fairness constraint with any given $\boldsymbol{\beta}$ can be obtained using the same algorithm proposed for ($\rm{P2}$) in this paper, with the throughput constraint in (\ref{Eq_CommonRateMaxConstraint1}) replaced by ${R_i}\left( {\bf{\tau }} \right) \ge {\beta _i}\bar R$, $i = 1,\,\, \cdots ,\,\,K$, where $\bar R$ here denotes the sum-throughput of all users.
\end{remark}

\section{Simulation Result}\label{SimulationResult}
In this section, we compare the maximum sum-throughput by ($\rm{P1}$) versus the maximum common-throughput by ($\rm{P2}$) in an example WPCN. The bandwidth is set as $1$MHz. It is assumed that the channel reciprocity holds for the DL and UL and thus $h_i = g_i$, $i=1,\, \cdots, \, K$, with the same pathloss exponent $\alpha_d = \alpha_u = \alpha$. Accordingly, both the DL and UL channel power gains are modeled as ${{h}_i} = {{g}_i} = 10^{-3}{\rho _{i}^2} { {D_i^{ - \alpha}}}$, $i=1,\, \cdots, \, K$, where $\rho_i$ represents the additional channel short-term fading which is assumed to be Rayleigh distributed, and thus $\rho_i^2$ is an exponentially distributed random variable with unit mean. Note that in the above channel model, a $30$dB average signal power attenuation is assumed at a reference distance of $1$m. The AWGN at the H-AP receiver is assumed to have a white power spectral density of $-160$dBm/Hz. For each user, the energy harvesting efficiency for WET is assumed to be $\zeta = 0.5$. Finally, we set $\Gamma = 9.8$dB assuming that an uncoded quadrature amplitude modulation (QAM) is employed \cite{Goldsmith}.

Fig. \ref{Fig_AvgThroughput} shows the maximum sum-throughput versus the maximum common-throughput in the same WPCN with $K = 2$, $D_1 = 5$m, and $D_2 = 10$m for different values of transmit power at H-AP, $P_A$, in dBm, by averaging over $1000$ randomly generated fading channel realizations, with fixed $\alpha = 2$. As shown in Fig. \ref{Fig_AvgThroughput}, when the sum-throughput is maximized, the throughput of $U_1$ dominates over that of $U_2$ due to the doubly near-far problem, which results in notably unfair rate allocation between the near user ($U_1$) and far user ($U_2$) in this example. It is also observed that the maximum common-throughput for the two users is smaller than the normalized maximum sum-throughput by the number of users, i.e., $R_{\rm{sum}} \left( \boldsymbol{\tau} \right)/K$ ($K=2$ in this example), which is a cost to pay in order to ensure a strictly fair rate allocation to the two users regardless of their distances from the H-AP.

\begin{figure}
   \centering
   \includegraphics[width=0.7\columnwidth]{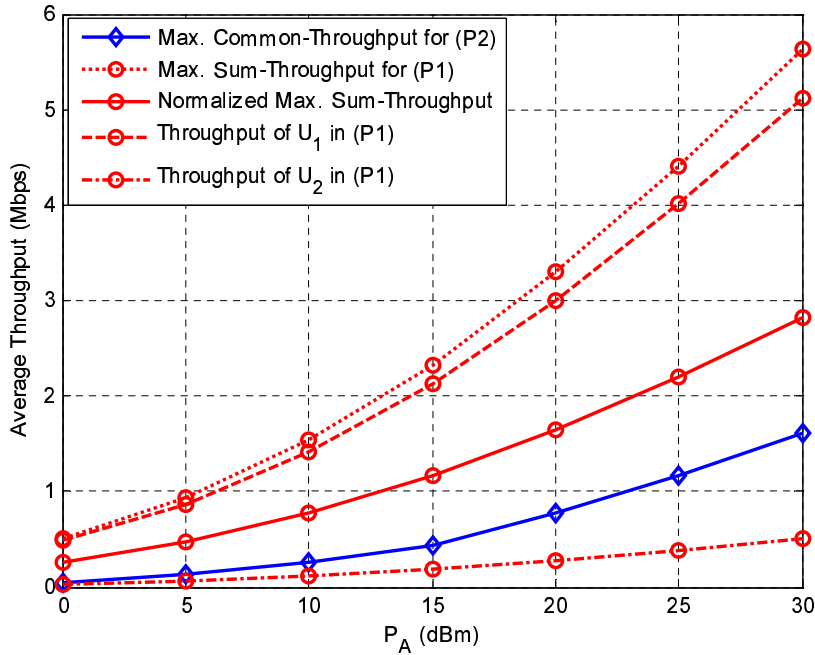}
   \caption{Sum-throughput vs. common-throughput. }
   \label{Fig_AvgThroughput}
\end{figure}

Next, by fixing $P_A = 20$dBm, Fig. \ref{Fig_Throughput_Normalized_by_MaxSumThroughput} shows the throughput comparison for different values of the common pathloss exponent $\alpha$ in both the DL and UL in the same WPCN as for Fig. \ref{Fig_AvgThroughput}. It is observed that when the sum-throughput is maximized, the throughput of the near user $U_1$ converges to the maximum sum-throughput as $\alpha$ increases, whereas that of the far user $U_2$ converges to zero, which indicates that the WPCN suffers from a more severe unfair rate allocation between the near and far users as the pathloss exponent increases, due to the doubly near-far problem. In addition, the maximum common-throughput for the two users is observed to decrease faster with increasing $\alpha$ than the normalized maximum sum-throughput. This is because as $\alpha$ increases, ($\rm{P2}$) allocates more time to the far user $U_2$ instead of near user $U_1$ in order to ensure the equal throughput allocation among users since the ratio $\gamma_1/\gamma_2$ increases with $\alpha$, whereas ($\rm{P1}$) allocates more time to $U_1$ instead of $U_2$ as $\alpha$ increases.

\begin{figure}
   \centering
   \includegraphics[width=0.7\columnwidth]{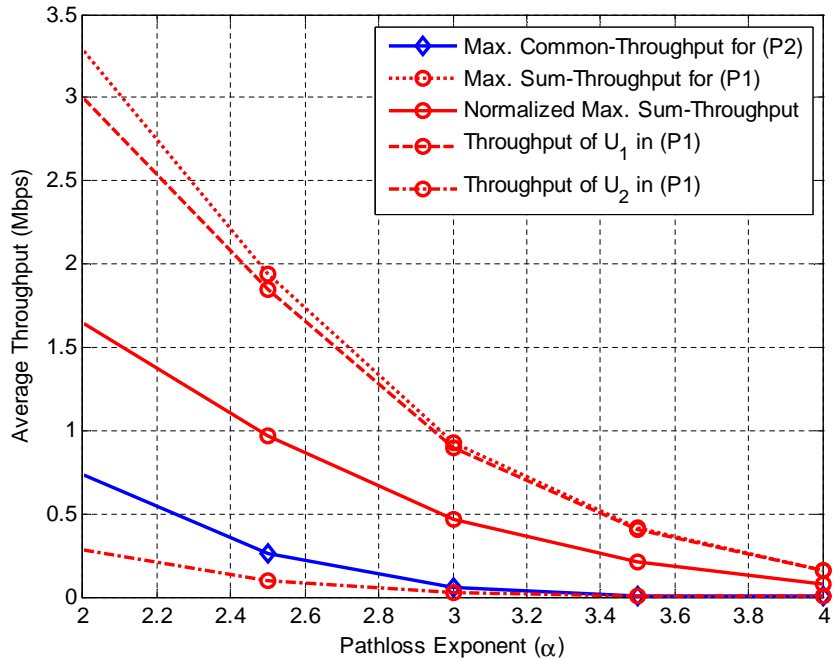}
   \caption{Throughput vs. pathloss exponent. }
   \label{Fig_Throughput_Normalized_by_MaxSumThroughput}
\end{figure}

At last, Fig. \ref{Fig_AvgThroughput_vs_NumUser} shows the throughput over number of users, $K$. It is assumed that $K$ users in the network are equally separated from the H-AP according to $D_i = \frac{D_K}{K} \times i$, $i=1, \, \cdots \, K$, where $D_K = 10$m. The transmit power at the H-AP and the pathloss exponent are set to be fixed as $P_A = 20$dBm and $\alpha = 2$, respectively. In addition, we compare with the throughput achievable by equal time allocation (ETA), i.e., $\tau_i = \frac{1}{K+1}$, $i=0, \, \cdots, \, K$, as a low-complexity time allocation scheme. It is observed that both the normalized maximum sum-throughput by solving (P1) and the maximum common-throughput by solving (P2) decreases with increasing $K$, and they outperform the sum-throughput and the minimum throughput (over all users) by the heuristic ETA scheme, respectively.

\begin{figure}
   \centering
   \includegraphics[width=0.7\columnwidth]{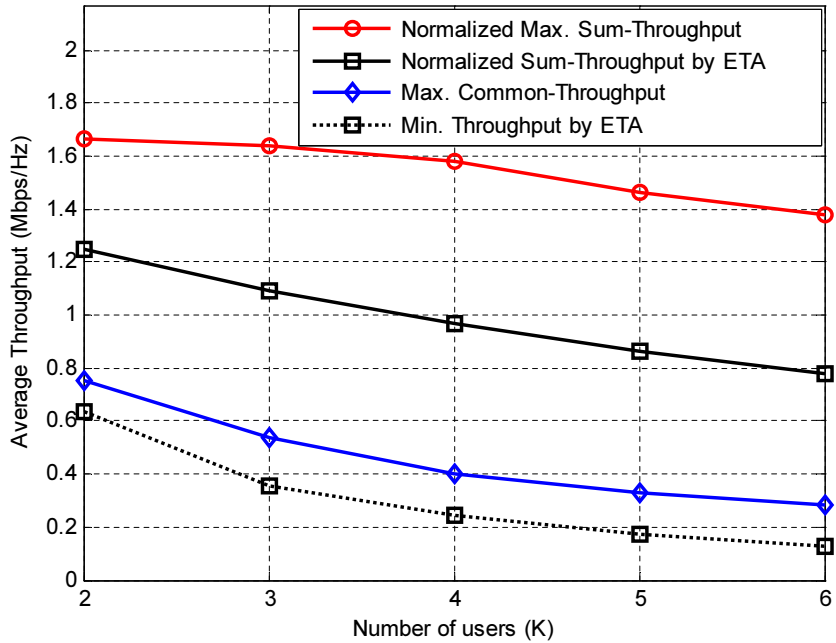}
   \caption{Throughput vs. number of users ($K$). }
   \label{Fig_AvgThroughput_vs_NumUser}
\end{figure}

\section{Conclusion} \label{Conclusion}
This paper has studied a new type of wireless RF (radio frequency) powered communication network with a harvest-then-transmit protocol, where the H-AP first broadcasts wireless energy to distributed users in the downlink and then the users transmit their independent information to the H-AP in the uplink by TDMA. Our results reveal an interesting new phenomenon in such hybrid energy-information transmission networks, so-called doubly near-far problem, which is due to the folded signal attenuation in both the downlink WET and uplink WIT. As a result, notably unfair time and throughput allocation among the users occurs when the conventional metric of network sum-throughput is maximized. To overcome this problem, we propose a new common-throughput maximization approach to allocate equal rates to all users regardless of their distances from the H-AP by allocating the transmission time to users inversely proportional to their distances to the H-AP. Simulation results showed that this approach is effective in solving the doubly near-far problem in the WPCN, but at a cost of sum-throughput degradation.

\appendices

   \section{Proof of Lemma \ref{Lemma_SumRateConcavity}}\label{App_Proof_Lemma_SumRateConcavity}
   Denote the Hessian of $R_i\left( \boldsymbol{\tau}  \right)$ defined in (\ref{Eq_Rate_Reference}) as
   \[
      {{\nabla ^2}R_i \left( \boldsymbol{\tau}  \right) = \left[ {{d_{j,m}^{(i)}}} \right], \,\,\, 0 \le j,m \le K,}
   \]
   where $d_{j,m}^{(i)}$ denotes the element of ${\nabla ^2}R_i\left( \boldsymbol{\tau}  \right)$ at the $j$th row and $m$th column. From (\ref{Eq_Rate_Reference}), the diagonal entries of ${\nabla ^2}R_i \left( \boldsymbol{\tau}  \right)$, i.e., $j = m$, can be expressed as
   \begin{equation}\label{App_D_nn}
      {{d_{j,j}^{(i)}} = \left\{ {\begin{array}{*{20}{c}}
      { - \frac{1}{{\ln 2}} {\gamma _i^2} {\tau _i ^ {-1}} \beta _i ^{-2} }  \\
      { - \frac{1}{{\ln 2}} {\gamma _i^2\tau _{0}^2} {\tau _i^{-3}}\beta_i ^{-2}}  \\
      0  \\
      \end{array}\begin{array}{*{20}{c}}
      {,\,\,\,\,\,\,\,\,\,\, j = 0 \,\,\,\,\,\,\,}  \\
      {,\,\,\,\,\,\,\,\,\,\, j = i \,\,\,\,\,\,\,\,\,}  \\
      {,\,\,\,{\rm{otherwise, \,\,}}}  \\
      \end{array}} \right.}
   \end{equation}
   where $\beta_i = 1 + \frac{{{\gamma _i}{\tau _{0}}}}{{{\tau _i}}}$. In addition, the off-diagonal entries of ${\nabla ^2}R_i \left( \boldsymbol{\tau}  \right)$ can be expressed as
   \begin{equation}\label{App_D_nm}
      {{d_{j,m}^{(i)}} = {d_{m,j}^{(i)}} = \left\{ {\begin{array}{*{20}{c}}
      { \frac{1}{{\ln 2}} {{\gamma _i^2}{\tau_0}}{\tau _i^{-2}} \beta_i^{-2} }  \\
      0  \\
      \end{array}\begin{array}{*{20}{c}}
      {,\,\,\,j = i\,\,\,{\rm{and}}\,\,\,m = 0}  \\
      {,\,\,\,\,\,\,\,\,\,\,\,  {\rm{otherwise.}} \,\,\,\,\,\,\,\,\,\, }  \\
      \end{array}} \right.}
   \end{equation}

   Given an arbitrary real vector ${\bf{v}} = [v_0, \,\, v_1, \,\, \cdots, \,\, v_K ]^T $, since $\tau_i \ge 0$, it can be shown from (\ref{App_D_nn}) and (\ref{App_D_nm}) that
   \[
      {{{\bf{v}}^T}{\nabla ^2}R_i\left( \boldsymbol{\tau}  \right){\bf{v}} = - \frac{1}{{\ln 2}}\frac{1}{\beta_i^2}\frac{{\gamma _i^2}}{{\tau _i^3}}{\left( {{v_i}{\tau _{0}} - {v_{0}}{\tau _i}} \right)^2}}
   \]
   \[
      {\,\, \le 0, \,\,\,\,\,\,\,\,\,\,\,\,\,\,\,\,\,\,\,\,\,\,\,\,}
   \]
   i.e., ${\nabla ^2}R_i \left( \boldsymbol{\tau}  \right)$, $\forall i$, is a negative semidefinite matrix. Therefore, $R_i \left( \boldsymbol{\tau}  \right)$ is a concave function of ${{\boldsymbol{\tau }}} = {\left[ {\tau _0\,\,\,\tau _1\,\,\, \cdots \,\,\,\tau _{N}} \right]^T}$ \cite{Boyd}. This completes the proof of Lemma \ref{Lemma_SumRateConcavity}.

   \section{Proof of Lemma \ref{Lemma_ExistenceOfSolution}}\label{App_Proof_Lemma_ExistenceOfSolution}
   From (\ref{Eq_Lemma_Existence}), we have
   \begin{equation}\label{App_Lemma_ZeroApproach}
      {\mathop {\lim }\limits_{z \to 0} f\left( z \right) = 1,}
   \end{equation}
   \begin{equation}\label{App_Lemma_1st_Derivative}
      {\frac{{\partial f\left( z \right)}}{{\partial z}} = \ln z,}
   \end{equation}
   \begin{equation}\label{App_Lemma_2nd_Derivative}
      {\frac{{{\partial ^2}f\left( z \right)}}{{\partial {z^2}}} = \frac{1}{z}.}
   \end{equation}
   Thus it follows from (\ref{App_Lemma_1st_Derivative}) and (\ref{App_Lemma_2nd_Derivative}) that $f \left(z \right)$ is a convex function over $z \ge 0$ and its minimum is attained at $z = 1$ with $f \left( 1 \right) = 0$. This implies that $f \left(z \right) \ge 0$ with $z \ge 1$ and is monotonically increasing with $z$ in this regime. Therefore, given $A > 0$, $f \left( z \right) = A$ has a unique solution $z^* > 1$. This completes the proof of Lemma \ref{Lemma_ExistenceOfSolution}.

   \section{Proof of Proposition \ref{Proposition_SumRateOptimalTime}}\label{App_Proof_Proposition_SumRateOptimalTime}
   The Lagrangian of $(\rm{P1})$ is given by
   \begin{equation}\label{App_P1_Lagrangian}
      {\mathcal{L}_{\rm{sum}}\left( {\boldsymbol{\tau} ,\nu } \right) = {{R_{\rm{sum}}}\left( \boldsymbol{\tau}  \right)}  - \nu \left( {\sum\limits_{i = 0}^K {{\tau _i}} } - 1 \right),}
   \end{equation}
   where $\nu \ge 0$ denotes the Lagrange multiplier associated with the constraint in (\ref{Eq_SumTime}). The dual function of ($\rm{P1}$) is thus given by
   \begin{equation}\label{App_DualFunction}
      {\mathcal{G}\left( \nu  \right) = \mathop {\min }\limits_{\boldsymbol{\tau}  \in \mathcal{D}} \,\,{\mathcal{L}_{{\rm{sum}}}}\left( {\boldsymbol{\tau} ,\nu } \right),}
   \end{equation}
   where $\mathcal{D}$ is the feasible set of $\boldsymbol{\tau}$ specified by (\ref{Eq_SumTime}) and (\ref{Eq_SumRateMaxConstraint2}).

   It can be shown from (\ref{Eq_SumTime}) and (\ref{Eq_SumRateMaxConstraint2}) that there exists an $\boldsymbol{\tau} \in \mathcal{D}$ with $\tau_i > 0$, $i = 0, \,\,1, \,\, \cdots \,\, K$, satisfying $\sum\limits_{i = 0}^K {{\tau _i}} < 1$, and thus strong duality holds for this problem thanks to the Slater's condition \cite{Boyd}. Since $(\rm{P1})$ is a convex optimization problem for which the strong duality holds, the Karush-Kuhn-Tucker (KKT) conditions are both necessary and sufficient for the global optimality of ($\rm{P1}$), which are given by
   \begin{equation}\label{App_OptimalityCondition1}
      {\sum\limits_{i = 0}^{K} {\tau _i^*}  \le 1,}
   \end{equation}
   \begin{equation}\label{App_OptimalityCondition2}
      {\nu^* \left( {\sum\limits_{i = 0}^K {{\tau _i^*}} } - 1 \right) = 0,}
   \end{equation}
   \begin{equation}\label{App_OptimalityCondition3}
      {\frac{\partial }{{\partial {\tau _i}}}R_{\rm{sum}} \left( \boldsymbol{\tau}^* \right) - \nu^* = 0, \,\,\,\, i = 0, \,\, \cdots, \,\,K,}
   \end{equation}
   where $\tau_i^*$'s and $\nu^*$ denote the optimal primal and dual solutions of ($\rm{P1}$), respectively. It can be easily verified that ${\sum\limits_{i = 0}^{K} {\tau _i^*} = 1}$ must hold for ($\rm{P1}$) and thus from (\ref{App_OptimalityCondition2}) without loss of generality, we assume $\nu^* > 0$. From (\ref{App_OptimalityCondition3}), it follows that
   \begin{equation}\label{App_OptimalityCondition4}
      {\sum\limits_{i = 1}^K {\frac{{{\gamma_i}}}{{1 + {\gamma_i}\frac{{\tau _{0}^*}}{{\tau _i^*}}}}} = \nu^* {\ln 2} ,}
   \end{equation}
   \begin{equation}\label{App_OptimalityCondition5}
      {  t \left( {\gamma_i}\frac{{\tau _{0}^*}}{{\tau _i^*}} \right) = \nu^*{\ln 2},\,\,\,\,\,1 \le i \le K,}
   \end{equation}
   where $t\left( x \right)$ is defined as
   \begin{equation}\label{App_MonotonicFunction}
      {t\left( x \right) \buildrel \Delta \over =  {\ln \left( {1 + x} \right) - \frac{x}{{1 + x}}}, \,\,\,\,\, x \ge 0.}
   \end{equation}

   Given $1 \le i,j \le K$, from (\ref{App_OptimalityCondition4}) we have
   \begin{equation}\label{App_Equality_in_MotonoticFunction}
      {t \left( {\gamma_i}\frac{{\tau _{0}^*}}{{\tau _i^*}} \right) = t \left( {\gamma_j}\frac{{\tau _{0}^*}}{{\tau _j^*}} \right), \,\,\,\,\, i \ne j.}
   \end{equation}
   It can be easily shown that $t\left( x \right)$ is a monotonically increasing function of $x \ge 0$ since $dt\left( x \right)/dx = x{\left( {1 + x} \right)^{ - 2}} \ge 0$ for $x \ge 0$. Therefore, equality in (\ref{App_Equality_in_MotonoticFunction}) holds if and only if ${\gamma_i}\frac{{\tau _{0}^*}}{{\tau _i^*}} = {\gamma_j}\frac{{\tau _{0}^*}}{{\tau _j^*}}$, $1 \le i,j \le K$, i.e.,
   \begin{equation}\label{App_TimePortionRatio}
      {\frac{{{\gamma_1}}}{{\tau _1^*}} = \frac{{{\gamma_2}}}{{\tau _2^*}} =  \cdots \frac{{{\gamma_K}}}{{\tau _K^*}} = C.}
   \end{equation}
   Note that $1 - \tau _{0}^* = \sum\limits_{j = 1}^K {\tau _j^*}$ and $\tau _j^* = \frac{{{\gamma_j}}}{{{\gamma_i}}}\tau _i^*$ from (\ref{App_OptimalityCondition1}) and (\ref{App_TimePortionRatio}), respectively. Therefore, $\tau _i^*$ can be expressed as
   \begin{equation}\label{App_Opt_Tau_n}
      {\tau _i^* = \left( {1 - \tau _{0}^*} \right)\frac{{{\gamma_i}}}{{\sum\limits_{j = 1}^K {\gamma_j} }} = \left( {1 - \tau _{0}^*} \right)\frac{{{\gamma_i}}}{A},}
   \end{equation}
   where $A = \sum\limits_{j = 1}^K {{\gamma_j}}$. In addition, it follows from (\ref{App_OptimalityCondition3}), (\ref{App_TimePortionRatio}), and (\ref{App_Opt_Tau_n}) that
   \begin{equation}\label{App_Equality2}
      {\ln \left( {1 + C\tau _{0}^*} \right) - \frac{{C\tau _{0}^*}}{{1 + C\tau _{0}^*}} = \frac{A}{{1 + C\tau _{0}^*}},\,\,\,\,\,1 \le i \le K,}
   \end{equation}
   where $C$ is defined in (\ref{App_TimePortionRatio}). Since $C = \frac{A}{{1 - \tau _{0}^*}}$ from (\ref{App_Opt_Tau_n}), we can modify (\ref{App_Equality2}) as
   \begin{equation}\label{App_Equality3}
      {z\ln z - z - A + 1 = 0,}
   \end{equation}
   where $z = 1 + \frac{{A\tau _{0}^*}}{{1 - \tau _{0}^*}}$. It is observed that $z > 1$ if $A > 0 $ and $0 < \tau _{0}^* < 1$. From Lemma \ref{Lemma_ExistenceOfSolution}, there exists a unique $z^* > 1$ that is the solution of (\ref{App_Equality3}). Therefore, the optimal time allocation to the DL WET is given by
   \begin{equation}\label{App_Opt_Tau_HAP}
      {\tau _{0}^* = \frac{{z^* - 1}}{{A + z^* - 1}}.}
   \end{equation}
   In addition, from (\ref{App_Opt_Tau_n}) and (\ref{App_Opt_Tau_HAP}), the optimal time allocation to the UL WITs, ${\tau _i^*}$, $1 \le i \le K$, is given by
   \begin{equation}\label{App_Opt_Tau_Sensors}
      {\tau _i^* = \frac{{{\gamma_i}}}{{A + z - 1}}.}
   \end{equation}
   This thus proves Proposition \ref{Proposition_SumRateOptimalTime}.

   \section{Proof of Corollary \ref{Corollary_SumThroughputMax_tau0}}\label{App_Proof_Corollary_SumThroughputMax_tau0}
   It can be easily shown from (\ref{Eq_Proposition_OptTimeAlloc}) that $\tau_0^* = 1$ with $A = 0$. From (\ref{App_Equality3}) and (\ref{App_Opt_Tau_HAP}), $\tau_0^*$ can be alternatively expressed as
   \begin{equation}\label{App_Eq_DL_TimeAlloc}
      {\tau_0^* = \frac{z^* - 1}{z^* \ln {z^*}}.}
   \end{equation}
   Given $A \ge 0$ and thus $z^* \ge 1$, both ${z^* \ln {z^*}}$ and $z^* - 1$ in (\ref{App_Eq_DL_TimeAlloc}) increase with $A$ since $z^*$ increases with $A$ as shown in the proof of Lemma \ref{Lemma_ExistenceOfSolution} given in Appendix \ref{App_Proof_Lemma_ExistenceOfSolution}. Furthermore, since $\frac{d}{d z} ({z \ln {z}}) = 1 + \ln {z}$ and $\frac{d}{dz}({z - 1}) = 1$, it follows that $\frac{d}{d z^*} ({z^* \ln {z^*}}) > \frac{d}{dz^*}({z^* - 1})$ with $z^* > 1$, i.e., ${z^* \ln {z^*}}$ increases faster with $z^*$ than $z^* - 1$. Therefore, it can be verified that ${z^* \ln {z^*}}$ increases faster with $A$ than $z^* - 1$, and thus $\tau_0^*$ decreases monotonically with increasing $A$. Finally, it can be shown that $\tau_0^* \to 0$ as $A \to \infty$ from the fact that $z^*$ increases with $A$.

   This thus completes the proof of Corollary \ref{Corollary_SumThroughputMax_tau0}.

   \section{Proof of Lemma \ref{Lemma_DualFuntion_Feasibility}}\label{App_Proof_Lemma_DualFunction_Feasibility}
   We first prove the ``if'' part of Lemma \ref{Lemma_DualFuntion_Feasibility}. If $\boldsymbol{\tau}' \in \mathcal{D}$ is a feasible solution for (\ref{Eq_FeasibilityProblem}) given $\bar R > 0$, i.e., $R_i \left( \boldsymbol{\tau}' \right) \ge \bar R$, $i=1, \,\, \cdots, \,\,K$, then for any $\boldsymbol{\lambda} \ge 0$ it follows from (\ref{Eq_Lagrangian}) that
   \[
      {\mathcal{G}} \left( \boldsymbol{\lambda} \right) \le {\mathcal{L}} \left( \boldsymbol{\tau}', \boldsymbol{\lambda} \right) \le 0,
   \]
   and thus $\mathop {\max }\limits_{\boldsymbol{\lambda}  \ge 0} \,\, \mathcal{G}\left( \boldsymbol{\lambda}  \right) \le 0$, which contradicts with the given assumption that there exists an ${\boldsymbol{\lambda}} \ge 0$ such that ${\mathcal{G}} \left( \boldsymbol{\lambda} \right) > 0$. The ``if'' part is thus proved.

   Next, we prove the ``only if'' part of Lemma \ref{Lemma_DualFuntion_Feasibility} by showing that its transposition is true, i.e, the problem in (\ref{Eq_FeasibilityProblem}) is feasible if ${\mathcal{G}} \left( \boldsymbol{\lambda} \right) \le 0$, $\forall {\boldsymbol{\lambda}} \ge 0$, by contradiction. Suppose that problem (\ref{Eq_FeasibilityProblem}) is feasible and there exits an ${\boldsymbol{\lambda}}'' \ge 0$ where ${\mathcal{G}} \left( \boldsymbol{\lambda}'' \right) > 0$. However, since (\ref{Eq_FeasibilityProblem}) is assumed to be feasible, there exists an $\boldsymbol{\tau}'' \in \mathcal{D}$ such that ${{R_i}\left( {\boldsymbol{\tau}}''  \right) \ge \bar R}$, $\forall i$, resulting in ${\lambda_i}'' \left( {{R_i}\left( {\boldsymbol{\tau}}'' \right) - \bar R} \right) \ge 0$ since $\boldsymbol{\lambda}'' \ge 0$. From (\ref{Eq_Lagrangian}) and (\ref{Eq_DualFunction}), we thus have
   \[
      \mathcal{G} \left( {\boldsymbol{\lambda}}'' \right) \le - \sum\nolimits_{i = 1}^K {{{\lambda_i}''}} \left( {{R_i}\left( {\boldsymbol{\tau}}  \right) - \bar R} \right) \le 0.
   \]
   This contradicts ${\mathcal{G}} \left( {\boldsymbol{\lambda}}'' \right) > 0$, and thus problem (\ref{Eq_FeasibilityProblem}) is feasible if ${\mathcal{G}} \left( \boldsymbol{\lambda} \right) \le 0$, $\forall {\boldsymbol{\lambda}} \ge 0$. The ``only if'' part is thus proved.

   Combining the above proofs of both ``if" and ``only if" parts, Lemma \ref{Lemma_DualFuntion_Feasibility} thus follows.

   \section{Proof of Proposition \ref{Proposition_WeightedSumRate}}\label{App_Proof_Lemma_DualFuntion_Feasibility}
   Given $\boldsymbol{\lambda} \ge 0$, the Lagrangian of problem (\ref{Eq_WSRProblem}) is given by
   \begin{equation}\label{App_WSR_Lagrangian}
      {\mathcal{L}_{\rm{WSR}}\left( {\boldsymbol{\tau} ,\mu } \right) = \sum\limits_{i = 1}^K {\lambda _i}{R_i}\left( \boldsymbol{\tau}  \right)  - \mu \left( {\sum\limits_{i = 0}^K {{\tau _i}} } - 1 \right),}
   \end{equation}
   where $\mu \ge 0$ denotes the Lagrange multiplier associated with the constraint in (\ref{Eq_SumTime}). The dual function of problem (\ref{Eq_WSRProblem}) is thus given by
   \begin{equation}\label{App_WSR_DualFunction}
      {\mathcal{G}_{\rm{WSR}}\left( \mu  \right) = \mathop {\min }\limits_{\boldsymbol{\tau}  \in \mathcal{D}} \,\,{\mathcal{L}_{{\rm{WSR}}}}\left( {\boldsymbol{\tau} ,\mu } \right).}
   \end{equation}
   Similar to ($\rm{P1}$), it can be easily shown that the problem in (\ref{Eq_WSRProblem}) is convex with zero duality gap. Therefore, the following KKT conditions must be satisfied by the optimal primal and dual solutions of problem (\ref{Eq_WSRProblem}):
   \begin{equation}\label{App_Eq_WSR_KKT1}
      {\ln \left( 1 + {\gamma _i}\frac{\tau _0^{\star}}{\tau _i^{\star}} \right) - \frac{{\gamma _i}\frac{\tau _0^{\star}}{\tau _i^{\star}}}{1 + {\gamma _i}\frac{\tau _0^{\star}}{\tau _i^{\star}}} = \frac{\mu^{\star}}{{{\lambda _i}}}\ln 2, \,\,\,\,\, i=1, \,\, \cdots, \,\,K,}
   \end{equation}
   \begin{equation}\label{App_Eq_WSR_KKT2}
      {\sum\limits_{i = 1}^K {\frac{{{\lambda _i}{\gamma_i}}}{1 + {\gamma _i}\frac{\tau _0^{\star}}{\tau _i^{\star}}}}  = {{{\mu}^{\star}}}\ln 2,}
   \end{equation}
   \begin{equation}\label{App_Eq_WSR_KKT3}
      {\sum\limits_{i = 0}^{K} {\tau _i^{\star}}  = 1,}
   \end{equation}
   where $\mu^{\star} > 0$ is the optimal dual solution. We then obtain (\ref{Eq_WSR_KKT1}) and (\ref{Eq_WSR_KKT2}) by changing variables as
   \begin{equation}\label{App_Eq_WSR_zi}
      {z_i} = {\gamma _i}\frac{\tau _0^{\star}}{\tau _i^{\star}}, \,\,\,\,\, i = 1, \,\, \cdots, K,
   \end{equation}
   in (\ref{App_Eq_WSR_KKT1}) and (\ref{App_Eq_WSR_KKT2}), respectively. It is worth noting that $z_1$ $\cdots$ $z_K$ and $\mu^*$ satisfying both (\ref{Eq_WSR_KKT1}) and (\ref{Eq_WSR_KKT2}) are uniquely determined since $K+1$ variables are solutions of $K+1$ independent equations and $\ln \left( {1 + {z_i}} \right) - \frac{{{z_i}}}{{1 + {z_i}}}$ is a monotonically increasing function of $z_i$. In addition, since $1 - {\tau _0^{\star}} = \sum\limits_{i = 1}^K {{\tau _i^{\star}}}$ from (\ref{App_Eq_WSR_KKT3}) and $\tau_i^{\star} = \tau_0^{\star} \frac{\gamma_i}{z_i}$ from (\ref{App_Eq_WSR_zi}), it follows that
   \begin{equation}\label{App_Eq_WSR_Solution1}
      {{\tau _0^{\star}}\left( {1 + \sum\limits_{i = 1}^K {\frac{{{\gamma_i}}}{{{z_i}}}} } \right) = 1,}
   \end{equation}
   from which we obtain (\ref{Eq_WSR_TimeAlloc_DL}). Finally, we obtain (\ref{Eq_WSR_TimeAlloc_UL}) from (\ref{App_Eq_WSR_zi}) and (\ref{App_Eq_WSR_Solution1}).

   This thus completes the proof of Proposition \ref{Proposition_WeightedSumRate}.


\begin{thebibliography}{1}
\bibliographystyle{IEEEbib}

   \bibitem{Zungeru}
      A. M. Zungeru, L. M. Ang, S. Prabaharan, and K. P. Seng, ``Radio frequency energy harvesting and management for wireless sensor networks,'' \emph{Green Mobile Devices and Netw.: Energy Opt. Scav. Tech.}, CRC Press, pp. 341-368, 2012.

   \bibitem{Vullers}
      R. J. M. Vullers, R. V. Schaijk, I. Doms, C. V. Hoof, and R. Mertens, ``Micropower energy harvesting,'' \emph{Elsevier Solid-State Circuits}, vol. 53, no. 7, pp. 684.693, July 2009.

   \bibitem{Shi}
      Y. Shi, L. Xie, Y. T. Hou, and H. D. Sherali, ``On renewable sensor networks with wireless energy transfer,'' in \emph{Proc. IEEE INFOCOM}, pp. 1350-1358, Oct. 2011.

   \bibitem{Huang}
      K. Huang and V. K. N. Lau, ``Enabling wireless power transfer in cellular networks: architecture, modeling and deployment,'' submitted for publication. (available on-line at arxiv:1207.5640)

   \bibitem{Lee}
      S. H. Lee, R. Zhang, and K. B. Huang, ``Opportunistic wireless energy harvesting in cognitive radio networks,'' \emph{IEEE Trans. Wireless Commun.}, vol. 12, no. 9, pp. 4788-4799, Sep. 2013.

   \bibitem{Varshney2008}
      L. R. Varshney, ``Transporting information and energy simultaneously,'' in \emph{Proc. IEEE Int. Symp. Inf. Theory (ISIT)}, pp. 1612-1616, July 2008.

   \bibitem{Grover}
      P. Grover and A. Sahai, ``Shannon meets Tesla: wireless information and power transfer,'' in \emph{Proc. IEEE Int. Symp. Inf. Theory (ISIT)}, pp. 2363-2367, June 2010.

   \bibitem{Zhang}
      R. Zhang and C. K. Ho, ``MIMO broadcasting for simultaneous wireless information and power transfer,'' \emph{IEEE Trans. Wireless Commun.}, vol. 12, no. 5, pp. 1989-2001, May 2013.

   \bibitem{Liu}
      L. Liu, R. Zhang, and K. C. Chua, ``Wireless information transfer with opportunistic energy harvesting,''  \emph{IEEE Trans. Wireless Commun.}, vol. 12, no. 1, pp. 288-300, Jan. 2013.

   \bibitem{Zhou}
      X. Zhou, R. Zhang, and C. K. Ho, ``Wireless information and power transfer: architecture design and rate-energy tradeoff,'' to appear in \emph{IEEE Trans. Commun.}. (available on-line at arXiv:1205.0618)

   \bibitem{Fouladgar}
      A. M. Fouladgar and O. Simeone, ``On the transfer of information and energy in multi-user systems,'' \emph{IEEE Commun. Letters}, vol. 16, no. 11, pp. 1733-1736, Nov. 2012.

   \bibitem{Boyd}
      S. Boyd and L. Vandenberghe, \emph{Convex Optimization}, Cambridge University Press, 2004.

   \bibitem{Knopp}
      R. Knopp and P. A. Humblet, ``Information capacity and power control in single-cell multi-user communications,'' in \emph{Proc. IEEE Int. Conf. Commun. (ICC)}, pp. 331-335, June, 1995.

   \bibitem{Tse}
      D. N. C. Tse and S. V. Hanly, ``Multiaccess fading channels-part I: polymatroid structure, optimal resource allocation and throughput capacities,'' \emph{IEEE Trans. Inf. Theory}, vol. 44, no. 7, pp. 2796-2815, Nov. 1998.

   \bibitem{Jindal}
      L. Li and A. J. Goldsmith, ``Capacity and optimal resource allocation for fading broadcast channels-part I: ergodic capacity,'' \emph{IEEE Trans. Inf. Theory}, vol. 47, no. 3, pp. 1083-1102, Mar. 2001.

   \bibitem{Zhang_MAC}
      R. Zhang, S. Cui, and Y. -C. Liang, ``On ergodic capacity of fading cognitive multiple-access and broadcast channels,'' \emph{IEEE Trans. Inf. Theory}, vol. 55, no. 11, pp. 5161-5178, Nov. 2009.

   \bibitem{LectureNote}
      S. Boyd, EE364b Lecture Notes. Stanford, CA: Stanford Univ., avaliable online at ${\rm{http://www}}{\rm{.stanford.edu/class/ee364b/lectures/}}$ ${\rm{ellipsoid\_method\_slides}}{\rm{.pdf}}$.

   \bibitem{Mohseni}
      M. Mohseni, R. Zhang, and J. M. Cioffi, ``Optimized transmission for fading multiple-access and broadcast channels with multiple antennas,'' \emph{IEEE J. Sel. Areas Commun.}, vol. 24, no. 8, pp. 1627-1639, Aug. 2006.

   \bibitem{Goldsmith}
      A. Goldsmith, \emph{Wireless Communications}, Cambridge University Press, 2005.


\clearpage



%\begin{table}\center\abovecaptionskip -0.6\baselineskip
%   \caption{Throughput of $U_2$ (kbps) vs. pathloss exponent}\center
%   \begin{tabular}{|c|c|c|}
%   \hline
%   \hline
%      $\alpha$  & Sum-thoughput apporach $\rm{(P1)}$ & Common-thoughput apporach $\rm{(P2)}$ \\
%   \hline
%      $2$       & 921.26                             & 175.78                                \\
%   \hline
%      $2.5$     & 331.73                             & 55.12                                 \\
%   \hline
%      $3$       & 74.31                              & 14.40                                 \\
%   \hline
%      $3.5$     & 10.62                              & 3.11                                  \\
%   \hline
%      $4$       & 1.26                               & 0.57                                  \\
%   \hline
%   \hline
%   \end{tabular}
%   \label{Table_Throughput_U2}
%\end{table}



\end{thebibliography}
\end{document}